\documentclass[acmsmall,screen]{acmart}

\usepackage{amsfonts}
\usepackage{mathrsfs}
\usepackage{amsmath}
\usepackage{bm}
\usepackage{mathtools}
\usepackage{mathpartir}
\usepackage{bbding}
\usepackage{color}
\usepackage[ruled,noend]{algorithm2e}
\usepackage[capitalize,nameinlink]{cleveref}
\usepackage{xparse}
\usepackage{ifthen}
\usepackage{thmtools}
\usepackage{thm-restate}
\usepackage{bm}
\usepackage{subcaption}
\usepackage{wrapfig}
\usepackage{tikz}
\usepackage{newfile}
\usepackage{totcount}

\definecolor{blue-green}{rgb}{0,0.5,0.5}

\regtotcounter{@todonotes@numberoftodonotes}

\usepackage{listings}

  \lstdefinestyle{tinyc}{
    basicstyle=\scriptsize\ttfamily,
    keywordstyle=\color{blue}
  }

  \lstdefinestyle{normalc}{
    basicstyle=\ttfamily,
    numbers=none,
    keywordstyle=\color{blue}
  }

  \lstdefinestyle{inlinec}{
    basicstyle=\ttfamily
  }

\usetikzlibrary{backgrounds}
\tikzstyle{every picture}+=[remember picture]

\usepackage{todonotes}

\newcommand{\RM}[1]{\todo[inline]{\textcolor{red}{\textbf{RM}: #1}}}

\usepackage{tikz}
\usetikzlibrary{automata}
\tikzset{gadget/.style={->,>=stealth,initial text=,minimum size=7pt,auto,on grid,scale=1,inner sep=1pt,node distance=1cm}}
\tikzset{every state/.style={minimum size=15pt,inner sep=1pt,fill=black!10,draw=black!70,thick}}

\usetikzlibrary{backgrounds,calc}

 \usetikzlibrary{ shapes, fit, arrows}
 \usetikzlibrary{decorations} 
\DeclareMathSymbol{\mdot}{\mathord}{symbols}{"01}
\usepackage{enumerate}
\usepackage{lscape}

\usepackage[utf8]{inputenc}
\usepackage[T1]{fontenc}
\usepackage{microtype}
\usepackage{enumitem}
\usepackage{bussproofs}
\usepackage{tikzit}

\tikzstyle{plain circle}=[fill=white, draw=black, shape=circle]
\tikzstyle{black circle}=[fill=black, draw=black, shape=circle]
\tikzstyle{blue circle}=[fill=blue, draw=black, shape=circle]
\tikzstyle{red circle}=[fill=red, draw=black, shape=circle]
\tikzstyle{blue-green circle}=[fill={rgb,255: red,0; green,128; blue,128}, draw=black, shape=circle]

\tikzstyle{direction edge}=[->]
\tikzstyle{dashed directional edge}=[->, dashed]
\tikzstyle{red full edge}=[draw=red, ->]
\tikzstyle{dashed simple edge}=[-, dashed]
\tikzstyle{double directional}=[double, ->]
\tikzstyle{dash double}=[dashed, double, ->]
\tikzstyle{new edge style 0}=[-, dashed, draw=red]
\tikzstyle{blue}=[-, draw=blue]
\tikzstyle{blue directiona}=[draw=blue, ->]
\tikzstyle{blue dashed}=[dashed, draw=blue, ->]

\def\sforall{{\mathbf{\forall\!\!\!\!\forall\,}}}
\def\sexists{{\mathbf{\raisebox{0.15ex}{$\exists$}\!\!\!\!\exists\,}}}

\DeclareDocumentCommand{\autstep}{O{}}{%
        \xrightarrow{#1}%
        }
\DeclareDocumentCommand{\langof}{O{} m}{%
  \mathsf{L}_{#1}(#2)%
  }

\newcommand{\pGCL}{\mathsf{pGCL}}
\newcommand{\pChoice}[1]{\oplus_{#1}}
\newcommand{\pChoiceWithoutParamters}{\oplus}
\newcommand{\nChoice}{\,[\!] \,}
\newcommand{\codeStyleMath}[1]{\mathtt{#1}}
\newcommand{\codeStyleText}[1]{\texttt{#1}}
\newcommand{\evaluationUnderContext}[2]{\llbracket #1 \rrbracket_{#2}}

\newcommand{\CK}{\omega_1^{\mathsf{CK}}}

\newcommand{\executionStateProjectionProbability}{\mathsf{Prob}}
\newcommand{\executionStateProjectionHistory}{\mathsf{Hist}}

\newcommand{\termProb}{\Pr_{\text{term}}}
\newcommand{\expRuntime}{\mathsf{ExpRuntime}}
\newcommand{\expRuntimeToReach}{\mathsf{ExpReachRuntime}}

\newcommand{\AST}{\mathsf{AST}}
\newcommand{\PAST}{\mathsf{PAST}}
\newcommand{\BAST}{\mathsf{BAST}}
\newcommand{\nAST}{\mathsf{AST}} 
\newcommand{\nPAST}{\mathsf{PAST}}
\newcommand{\nbPAST}{\mathsf{BAST}}

\newcommand{\nffAST}{\mathsf{Fin\text{-}Fair\text{-}AST}}
\newcommand{\nffPAST}{\mathsf{Fin\text{-}Fair\text{-}PAST}}
\newcommand{\nffBAST}{\mathsf{Fin\text{-}Fair\text{-}BAST}}

\newcommand{\nfPAST}{\mathsf{Fair\text{-}PAST}}
\newcommand{\nfBAST}{\mathsf{Fair\text{-}BAST}}
\newcommand{\nfAST}{\mathsf{Fair\text{-}AST}}

\newcommand{\setOfReals}{\mathbb{R}}

\newcommand{\setOfRationals}{\mathbb{Q}}
\newcommand{\setOfPositiveRationals}{\mathbb{Q}^+}
\newcommand{\setOfNaturals}{\mathbb{N}}

\newcommand{\setOfVariables}{\mathsf{Var}}
\newcommand{\setOfVariableValuations}{\mathbb{V}}
\newcommand{\setOfSchedules}{\mathbb{F}}
\newcommand{\setOfExecutionStates}{\mathbb{P}}
\newcommand{\setOfProgramStates}{\Sigma}
\newcommand{\setOfReachableStates}[1]{\Sigma_r[#1]}

\newcommand{\setOfPrograms}{\codeStyleMath{Prog}}
\newcommand{\setOfPartialSchedulesOfSize}[1]{\mathcal{F}_{#1}} 
\newcommand{\setOfPartialSchedules}{\mathcal{F}}

\newcommand{\setOfRecursiveWellFoundedTrees}{\Omega_{rec}} 

\newcommand{\listingsInLatex}{Code}

\newcommand{\listingLineRef}[1]{Line ~\ref{#1}}

\newcommand{\codeRef}[1]{Program \ref{#1}} 
\newcommand{\declareCodeFigure}{\renewcommand\figurename{Prg.}} 

\newcommand{\finitaryTransformation}[1]{\codeStyleMath{fin(}#1\codeStyleMath{)}}
\newcommand{\fairnessPredicate}[1]{\mathsf{fair(}#1\mathsf{)}}
\newcommand{\fairnessPredicateWithoutParameters}{\mathsf{fair}}
\newcommand{\fairnessBranchPredicate}[1]{\mathsf{fair_{br}(}#1\mathsf{)}}
\newcommand{\fairnessBranchPredicateWithoutParameters}{\mathsf{fair_{br}}}

\newcommand{\initialExecutionState}[1]{\sigma^e_{#1, 0}} 

\newcommand{\initialProgramState}[1]{\sigma_{#1, 0}}

\newcommand{\terminalStateSet}{T_e}
\newcommand{\terminalStateSetInStep}[1]{T^{#1}_e}

\newcommand{\oneStepExec}[1]{\vdash_{#1}} 
\newcommand{\nStepExec}[2]{\vdash^{#2}_{#1}} 
\newcommand{\starStepExec}[1]{\vdash^*_{#1}}

\newcommand{\nameForNormalFormPrograms}{Knievel } 

\newcommand{\setOfLowerStates}[1]{\mathsf{Lower}_{#1}}

\newcommand{\ordinalVariableSymbol}{\mathbf{o}}
\newcommand{\ordMath}{\text{ord}}
\newcommand{\node}{\mathsf{node}}
\newcommand{\INC}{\mathsf{inc}}

\newcommand{\calM}{\mathcal{M}}
\def\sQ{\mathbf{Q}}

\title{Positive Almost-Sure Termination -- Complexity and Proof Rules}

\newcommand{\OurInstitution}{Max Planck Institute for Software Systems (MPI-SWS)}
\newcommand{\OurStreet}{Paul-Ehrlich-Stra{\ss}e, Building G26}
\newcommand{\OurCity}{Kaiserslautern}
\newcommand{\OurPostcode}{67663}
\newcommand{\OurCountry}{Germany}

\author{Rupak Majumdar}
\orcid{0000-0003-2136-0542}             
\affiliation{
  \institution{\OurInstitution}            
  \streetaddress{\OurStreet}
  \city{\OurCity}
  \postcode{\OurPostcode}
  \country{\OurCountry}                    
}

\email{rupak@mpi-sws.org}          

\author{V.R.~Sathiyanarayana}
\orcid{0009-0006-5187-5415}             
\affiliation{
  \institution{\OurInstitution}            
  \streetaddress{\OurStreet}
  \city{\OurCity}
  \postcode{\OurPostcode}
  \country{\OurCountry}                    
}
\email{sramesh@mpi-sws.org}          

\begin{abstract}
We study the recursion-theoretic complexity of Positive Almost-Sure Termination ($\PAST$) in an imperative
programming language with rational variables, bounded nondeterministic choice, and discrete
probabilistic choice.
A program terminates positive almost-surely if, for every scheduler, the program terminates almost-surely 
and the expected runtime to termination is finite.
We show that $\PAST$ for our language is complete for the (lightface) co-analytic sets ($\Pi_1^1$-complete).
This is in contrast to the related notions of Almost-Sure 
Termination ($\AST$) and Bounded Termination ($\BAST$), both of which are arithmetical ($\Pi^0_2$- and $\Sigma^0_2$-complete respectively). 

Our upper bound implies an effective procedure to reduce reasoning about probabilistic termination to non-probabilistic
fair termination in a model with bounded nondeterminism, and to simple program termination in models with unbounded nondeterminism.
Our lower bound shows the opposite: for every program with unbounded nondeterministic choice, there is an effectively computable probabilistic program with
bounded choice such that the original program is terminating $iff$ the transformed program is $\PAST$.

We show that every program has an effectively computable normal form,
in which each probabilistic choice either continues or terminates execution immediately, each with probability $1/2$.
For normal form programs, we provide a sound and complete proof rule for $\PAST$. 
Our proof rule uses transfinite ordinals.
We show that reasoning about $\PAST$ requires transfinite ordinals up to $\CK$; 
thus, existing techniques for probabilistic termination
based on ranking supermartingales that map program states to reals 
do not suffice to reason about $\PAST$.
\end{abstract}

\begin{document}

\makeatletter
\newtheoremstyle{dazzle}%
{.5\baselineskip\@plus.2\baselineskip
  \@minus.2\baselineskip}
{.5\baselineskip\@plus.2\baselineskip
  \@minus.2\baselineskip}
{\@acmplainbodyfont}
{\@acmplainindent}
{\bfseries}
{.}
{.5em}
{\thmname{\textcolor{red}{\textbf{#1}}}\thmnumber{ \textcolor{red}{\textbf{#2}}}\thmnote{ {\@acmplainnotefont(\textcolor{blue}{#3})}}}
\makeatother

\theoremstyle{dazzle}
\newtheorem{maintheorem}[theorem]{Theorem}
\newtheorem{mainlemma}[theorem]{Lemma}
\newtheorem{maincorollary}[theorem]{Corollary}
\newtheorem{mainproposition}[theorem]{Proposition}

\crefalias{maintheorem}{theorem}
\crefalias{mainlemma}{lemma}
\crefalias{maincorollary}{corollary}
\crefalias{mainproposition}{proposition}

\theoremstyle{acmplain}
\newtheorem{observation}[theorem]{Observation}
\theoremstyle{acmdefinition}
\newtheorem{remark}[theorem]{Remark}

\Crefname{observation}{Observation}{Observations}

\renewcommand{\lstlistingname}{\listingsInLatex}

\maketitle

\section{Introduction}
\label{sec:intro}

A \emph{probabilistic} program augments an imperative program with primitives for randomization.
Probabilistic programs allow direct implementation of randomized computation and probabilistic modeling
and have found applications in machine learning, bio-informatics, epidemiology, and information retrieval amongst others; see \citet{KatoenGJKO15} for a comprehensive presentation of their applicability.

We study programs written in a classical imperative language with constructs for bounded (binary) nondeterministic choice
$P_1 \nChoice P_2$ and discrete probabilistic choice $P_1 \pChoice{p} P_2$.
The first program can nondeterministically reduce  to either $P_1$ or $P_2$; 
the second reduces to $P_1$ with probability $p$ and to $P_2$ with probability $1-p$. 

A fundamental and classical question about programs is \emph{termination}: does the execution of a program stop
after a finite number of steps?
In the presence of nondeterministic choice, a program can have many executions, depending on how the nondeterminism is resolved.
Typically, nondeterminism is modelled as being resolved demonically by an uncaring \emph{scheduler}, and 
the termination question is modified to ask: 
does the program stop after a finite number of steps no matter how the scheduler resolves nondeterminism?

If, in addition, a program has probabilistic choice, the notion of termination has to be modified to exclude 
some ostensibly infinite executions with a total measure of zero.
For example, if a program repeatedly tosses a fair coin until it lands heads, it will halt with probability one, as the probability of observing an
infinite sequence of tails is zero.

Consequently, several qualitative notions of termination have been defined and studied.
A program is \emph{almost sure terminating}, written $\AST$, if for every scheduler, the probability of termination is one.
A program is \emph{positive} almost sure terminating, written $\PAST$, if for every scheduler, the expected run time to termination is finite.
Finally, a program is \emph{bounded} almost sure terminating, written $\BAST$, if there exists a global bound on expected run times to termination independent of the scheduler.

Clearly, every $\BAST$ program is also $\PAST$, and every $\PAST$ program is also $\AST$.
In the absence of nondeterminism, $\PAST$ and $\BAST$ coincide.
However, these notions are different in general, as illustrated in Programs \ref{fig:example-program-1} and \ref{fig:example-program-2}.

\begin{figure}
  \declareCodeFigure
  \small
  \begin{subfigure}{0.45\textwidth}
    \begin{lstlisting}[language=python, mathescape=true, escapechar=|, xleftmargin=15pt]
x $\coloneqq$ 1
while (x |$\neq$| 0):
  x $\coloneqq$ x + 1 |$\pChoice{\frac{1}{2}}$| x $\coloneqq$ x - 1
    \end{lstlisting}
    \caption{This program is $\AST$ but not $\PAST$.}
    \label{fig:example-program-1}
  \end{subfigure}
  \begin{subfigure}{0.45\textwidth}
    \begin{lstlisting}[language=python, mathescape=true, escapechar=|, xleftmargin=15pt]
x, y, z $\coloneqq$ 0, 0, 1
while (x + y = 0): |\label{line:example-npast-loop-one}|
  y $\coloneqq$ 0 |$\nChoice$| y $\coloneqq$ 1
  x $\coloneqq$ 0 |$\pChoice{\frac{1}{2}}$| x $\coloneqq$ 1
  z $\coloneqq$ z * 4 |\label{line:example-npast-inc}|
while (x = 0 |$\land$| z > 0):
  z $\coloneqq$ z - 1
    \end{lstlisting}
    \caption{This program is $\PAST$ but not $\BAST$.}
    \label{fig:example-program-2}
  \end{subfigure}
  \caption{Programs showcasing the relationships between $\AST$, $\PAST$, and $\BAST$.}
\end{figure}
\codeRef{fig:example-program-1} is the famous symmetric random walker, which terminates almost surely (i.e., is $\AST$) but cannot expect to do so in a finite amount of time \cite{Polya}.
Meanwhile, \codeRef{fig:example-program-2} is $\PAST$, but
the longer the scheduler keeps the execution inside the first loop (from Lines ~\ref{line:example-npast-loop-one} to ~\ref{line:example-npast-inc}),
the greater its expected runtime.
Thus, it is not $\BAST$.
However, replacing \listingLineRef{line:example-npast-inc} by \lstinline{z = z + 1} induces an upper bound of $4$ over all possible expected runtimes, making it $\BAST$.


All these notions have been studied extensively, both with and without (demonic) nondeterminism \cite{Pnueli83,McIverM05,McIverMKK18,BournezG05,FuC19}.
One main focus of these works has been the development of proof rules to prove that a given program terminates under one of these notions.
Most of this work has focused on $\AST$ and $\BAST$; relatively little is known for $\PAST$.
%

In this paper, we characterize the recursion-theoretic complexity of $\PAST$ and provide a semantically sound and complete
proof rule.
Our first result is that membership in $\PAST$ is complete for the (lightface) co-analytic sets, that is, $\Pi_1^1$-complete.
This is in contrast to $\AST$ and $\BAST$, both of which lie in the arithmetic hierarchy ($\Pi^0_2$-complete and $\Sigma^0_2$-complete,
respectively \cite{KaminskiKM19}). 
Hardness already holds with binary nondeterministic choice and probabilistic choice of the form 
\small\begin{equation}\label{eq:gnf}
    {
\codeStyleMath{skip}\ \oplus_{1/2} \codeStyleMath{exit}}
\tag{\nameForNormalFormPrograms form}
\end{equation}\normalsize
which continues execution or halts with probability $1/2$ each.
A consequence of our result is that every probabilistic program has an effectively constructible \emph{normal form}, 
which we call \emph{\nameForNormalFormPrograms form} (after Evel Knievel, who made many such choices in his life).
Our second main result is a sound and complete proof rule for \nameForNormalFormPrograms form $\PAST$ programs.
We prove that proof systems for $\PAST$ require transfinite ordinals up to the first non-computable ordinal $\CK$, also known as the Church-Kleene ordinal.
This is in contrast to $\AST$ and $\BAST$, neither of which require transfinite ordinals.
In fact, most proof systems for $\AST$ and $\BAST$ use \emph{ranking supermartingales} that map program states to the reals with the proviso that
each program transition decreases the expected value of the mapping by a minimum amount \cite{FuC19,FioritiH15,ChakarovS13}.
Our result shows that such an attempt will not work for $\PAST$.
To illustrate this claim, we describe in \cref{subsec:hydra-example} a stochastic variant of the \emph{Hydra} game \cite{KirbyParis1982} that shows an intuitive example of a $\PAST$ program that
requires transfinite ordinals up to $\varepsilon_0$ to demonstrate termination.
Recall that the complexity of valid statements in the standard model of arithmetic is $\Delta^1_1$ \cite{Rogers};
thus, relative completeness results for $\PAST$ must use more powerful proof systems.

Our $\PAST$ proof rule for \nameForNormalFormPrograms form programs uses two ingredients.
The first is a \emph{ranking function} from program states to ordinals up to $\CK$ with the property that only terminal states are ranked zero.
The second is a state-dependent certificate, based on ranking supermartingales, for a bound on the expected time to reach a state
with a lower rank independent of the scheduler.

We show that for every program---not necessarily in \nameForNormalFormPrograms form---the proof rule is complete: from every $\PAST$ program, one can extract a rank and a certificate.
Moreover, by analyzing the possible traces of programs in \nameForNormalFormPrograms form, we show that the rule is sound: the existence of such a ranking function and a ranking supermartingale
implies that the expected running time is bounded for each scheduler.
However, soundness depends on the normal form: the rule is not sound if applied to general programs.
Since our first result provides an effective transformation to \nameForNormalFormPrograms form, 
we nevertheless get a semantically sound and complete proof system by first transforming the program into the normal
form and then applying the proof rule.

We also show that ordinals up to $\CK$ are necessary by explicitly constructing, for each constructible ordinal $\mathsf{o} < \CK$, a $\PAST$ program for which suitable
ranking functions include $\mathsf{o}$ in their range.
Our construction encodes a recursive $\omega$-tree $T$ into a probabilistic program $P(T)$ such that $T$ is well-founded \emph{iff} $P(T)$ is $\PAST$---recall
that the constructible ordinals are coded by such trees \cite{Kozen06}.

Our results are related to termination and fair termination problems for non-probabilistic programs with unbounded countable nondeterministic choice \cite{Chandra78,HarelK84,AptP86,Harel86}.
The $\Pi^1_1$-completeness and the requirement of ordinals up to $\CK$ for deciding termination of programs with countable nondeterministic choice was shown by \citet{Chandra78} and \citet{AptP86}.
Additionally, \citet{Harel86} showed a general recursive transformation on trees with bounded nondeterministic choice and fairness that reduces fair termination to termination, thereby providing
a semantically complete proof system for fair termination.
Since fairness can simulate countable nondeterminism using bounded nondeterminism, these results also show a lower complexity bound and the necessity of transfinite ordinals for fair termination.
Our results show that countable nondeterminism and discrete probabilistic choice has the same power.

We summarize our main results below:
\begin{enumerate}[topsep=0pt]
\item Deciding if a probabilistic program with bounded nondeterministic and probabilistic choice is $\PAST$ is $\Pi^1_1$-complete. 
\item For any probabilistic program $P$, there is an effectively constructible \nameForNormalFormPrograms form program $P_K$
and non-probabilistic program $P_1$ with bounded nondeterministic choice and non-probabilistic program $P_2$ with unbounded choice 
such that $P$ is $\PAST$ \emph{iff} $P_K$ is $\PAST$ \emph{iff} $P_1$ is fairly terminating \emph{iff} $P_2$ is terminating.

\item For any recursive $\omega$-tree $T$, there is a probabilistic program $P(T)$ such that $T$ is well-founded \emph{iff} $P(T)$ is $\PAST$.
Hence, proving $\PAST$ requires ordinals up to $\CK$.

\item There is a sound and complete proof rule for \nameForNormalFormPrograms form programs that uses a (deterministic) ranking function with codomain $\CK$
and ranking supermartingales.
While the rule is complete for every $\PAST$ program, it is only sound for programs in \nameForNormalFormPrograms form.
\end{enumerate}

\section{A Hydra Game: $\PAST$ Requires Transfinite Ordinals}
\label{subsec:hydra-example}


We now illustrate our main arguments in a stochastic variant of the \emph{Hydra game}, 
a two player game between the warrior Hercules and the Lernaean Hydra.
Introduced by \citet{KirbyParis1982}, the deterministic Hydra game terminates but requires transfinite ordinals to prove as much.
Our stochastic version is $\PAST$ and similarly requires transfinite ordinals to prove its membership.
\begin{figure}[t]
    \declareCodeFigure
\smaller
\begin{lstlisting}[
%caption={The Augmented Hydra game}, captionpos=b, label={lst:hydra-game}, 
	language=python, mathescape=true, escapechar=|, xleftmargin=15pt]
n $\coloneqq$ 4	# initial regrowth capacity
while (True):
    if (empty(hydra)): exit # Hercules has killed the Hydra
    l $\coloneqq$ Hercules(hydra) # Hercules's choice |\label{line:hercules-move}|
    parent $\coloneqq$ getParent(hydra, l) # the parent node |\label{line:parent-op1}|
    grandParent $\coloneqq$ getParent(hydra, parent) # the grandparent node |\label{line:parent-op2}|
    hydra $\coloneqq$ removeLeaf(hydra, l) |\label{line:remove-head}| # disconnect head

    if (not empty(grandparent)): # grow new heads
        evolve $\coloneqq$ 0 |$\nChoice$| evolve $\coloneqq$ 1 # Hydra's move |\label{line:hydra-move1}|
        while(evolve):
            skip |$\pChoice{1/2}$| exit |\label{line:death-chance2}| # die with some probability
            n $\coloneqq$ n * 4 # quadruple regrowth capacity
            evolve $\coloneqq$ 0 |$\nChoice$| evolve $\coloneqq$ 1 # Hydra's move |\label{line:hydra-move2}|
        subtree $\coloneqq$ getSubtree(hydra, parent) |\label{line:subtree-shape}| # find the place to grow heads
        hydra $\coloneqq$ growNewHeads(n - 1, grandparent, subtree) |\label{line:head-growth}| # grow new heads
\end{lstlisting}
\caption{The Hydra Game. $s_1 \pChoice{1/2} s_2$ is a probabilistic choice between statements $s_1$ and $s_2$: 
the program transitions to $s_1$ or $s_2$ with probability $1/2$ each. 
$s_1 \nChoice s_2$ is a nondeterministic choice: the program transitions to $s_1$ or $s_2$ nondeterministically.}
\label{fig:hydra-game}
\vspace{-12pt}
\end{figure}

Just like the original \cite{KirbyParis1982}, our stochastic variant is a two player-game between Hercules and the Hydra. 
The Hydra is a finite rooted tree.
A \emph{head} of the Hydra is a leaf together with the edge connecting the leaf to the tree. Naturally, the Hydra can have multiple heads.

Each round of the game begins with Hercules chopping off one of the Hydra's heads.
In the traditional game, the Hydra responds by growing two new heads and ending the round.
Our variant is a little different.
First, our game maintains a number \codeStyleText{n}, initially \codeStyleText{4}, that measures the Hydra's head growth capabilities.
Additionally, our Hydra can (try to) improve its chances by attempting to \emph{evolve} several times. 
Evolution is risky: with probability $1/2$, it causes the Hydra to implode, instantly ending the game in Hercules' favour.
However, if successful, it quadruples the Hydra's growth capacity \codeStyleText{n}.
After (possibly many) successful evolution(s), the Hydra instantly grows new heads in the following way:
if the grandparent node \codeStyleText{grandParent} of the leaf chopped of by Hercules exists,
$\codeStyleMath{n} - 1$ smaller hydras are grown beneath \codeStyleText{grandParent}, with 
each baby hydra taking the shape of the remaining subtree rooted at the parent of the leaf that was chopped off.
Hercules now picks and chops off another head, moving the game onto its next round.

The game is described in greater detail in \codeRef{fig:hydra-game}.
See \cref{fig:hydra-game-representation} for an illustration of a move.
Notice that the Hydra cannot evolve or grow new heads if the leaf removed by Hercules had no grandparent. 
%
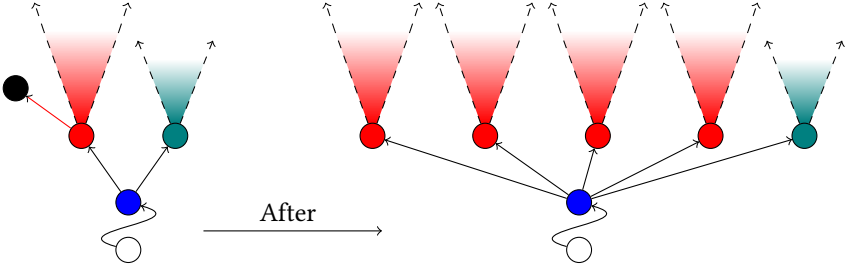
\begin{figure}[t]
    \scalebox{1}{\begin{tikzpicture}
	\begin{pgfonlayer}{nodelayer}
		\node [style=plain circle] (0) at (-4.5, 0) {};
		\node [style=blue circle] (1) at (-4.5, 1.25) {};
		\node [style=red circle] (3) at (-5.75, 3) {};
		\node [style=black circle] (10) at (-7.5, 4.25) {};
		\node [style=none] (20) at (-2.5, 0.5) {};
		\node [style=none] (21) at (2.25, 0.5) {};
		\node [style=none] (22) at (-0.25, 1) {After};
		\node [style=none] (59) at (-6.75, 5.75) {};
		\node [style=none] (60) at (-4.75, 5.75) {};
		\node [style=none] (80) at (-7, 6.5) {};
		\node [style=none] (81) at (-4.5, 6.5) {};
		\node [style=red circle] (82) at (2, 3) {};
		\node [style=none] (83) at (1, 5.75) {};
		\node [style=none] (84) at (3, 5.75) {};
		\node [style=none] (85) at (0.75, 6.5) {};
		\node [style=none] (86) at (3.25, 6.5) {};
		\node [style=red circle] (87) at (5, 3) {};
		\node [style=none] (88) at (4, 5.75) {};
		\node [style=none] (89) at (6, 5.75) {};
		\node [style=none] (90) at (3.75, 6.5) {};
		\node [style=none] (91) at (6.25, 6.5) {};
		\node [style=red circle] (92) at (8, 3) {};
		\node [style=none] (93) at (7, 5.75) {};
		\node [style=none] (94) at (9, 5.75) {};
		\node [style=none] (95) at (6.75, 6.5) {};
		\node [style=none] (96) at (9.25, 6.5) {};
		\node [style=red circle] (97) at (11, 3) {};
		\node [style=none] (98) at (10, 5.75) {};
		\node [style=none] (99) at (12, 5.75) {};
		\node [style=none] (100) at (9.75, 6.5) {};
		\node [style=none] (101) at (12.25, 6.5) {};
		\node [style=plain circle] (102) at (7.5, 0) {};
		\node [style=blue circle] (103) at (7.5, 1.25) {};
		\node [style=blue-green circle] (104) at (-3.25, 3) {};
		\node [style=none] (105) at (-4.25, 5.5) {};
		\node [style=none] (106) at (-2.25, 5.5) {};
		\node [style=blue-green circle] (107) at (13.5, 3) {};
		\node [style=none] (108) at (12.5, 5.5) {};
		\node [style=none] (109) at (14.5, 5.5) {};
	\end{pgfonlayer}
	\begin{pgfonlayer}{edgelayer}
		\draw [style=direction edge] (1) to (3);
		\draw [style=direction edge, in=-15, out=165, looseness=3.25] (0) to (1);
		\draw [style=red full edge] (3) to (10);
		\draw [style=direction edge] (20.center) to (21.center);
		\draw [style=dashed simple edge] (3) to (59.center);
		\draw [style=dashed simple edge] (3) to (60.center);
		\draw [style=dashed directional edge] (59.center) to (80.center);
		\draw [style=dashed directional edge] (60.center) to (81.center);
		\draw [style=dashed simple edge] (82) to (83.center);
		\draw [style=dashed simple edge] (82) to (84.center);
		\draw [style=dashed directional edge] (83.center) to (85.center);
		\draw [style=dashed directional edge] (84.center) to (86.center);
		\draw [style=dashed simple edge] (87) to (88.center);
		\draw [style=dashed simple edge] (87) to (89.center);
		\draw [style=dashed directional edge] (88.center) to (90.center);
		\draw [style=dashed directional edge] (89.center) to (91.center);
		\draw [style=dashed simple edge] (92) to (93.center);
		\draw [style=dashed simple edge] (92) to (94.center);
		\draw [style=dashed directional edge] (93.center) to (95.center);
		\draw [style=dashed directional edge] (94.center) to (96.center);
		\draw [style=dashed simple edge] (97) to (98.center);
		\draw [style=dashed simple edge] (97) to (99.center);
		\draw [style=dashed directional edge] (98.center) to (100.center);
		\draw [style=dashed directional edge] (99.center) to (101.center);
		\draw [style=direction edge, in=-15, out=165, looseness=3.25] (102) to (103);
		\draw [style=direction edge] (103) to (82);
		\draw [style=direction edge] (103) to (87);
		\draw [style=direction edge] (103) to (92);
		\draw [style=direction edge] (103) to (97);
		\draw [style=direction edge] (1) to (104);
		\draw [style=dashed directional edge] (104) to (105.center);
		\draw [style=dashed directional edge] (104) to (106.center);
		\draw [style=dashed directional edge] (107) to (108.center);
		\draw [style=dashed directional edge] (107) to (109.center);
		\draw [style=direction edge] (103) to (107);
	\end{pgfonlayer}
	\begin{pgfonlayer}{background}
		\shade[top color=white, bottom color=red] (3) -- (59.center) -- (60.center) -- (3);
		\shade[top color=white, bottom color=red] (82) -- (83.center) -- (84.center) -- (82);
		\shade[top color=white, bottom color=red] (87) -- (88.center) -- (89.center) -- (87);
		\shade[top color=white, bottom color=red] (92) -- (93.center) -- (94.center) -- (92);
		\shade[top color=white, bottom color=red] (97) -- (98.center) -- (99.center) -- (97);
		\shade[top color=white, bottom color=blue-green] (104) -- (-3.25 * 0.2 + -4.25 * 0.8, 3 * 0.2 + 5.5 * 0.8) -- (-3.25 * 0.2 + -2.25 * 0.8, 3 * 0.2 + 5.5 * 0.8) -- (104);
		\shade[top color=white, bottom color=blue-green] (107) -- (13.5 * 0.2 + 12.5 * 0.8, 3 * 0.2 + 5.5 * 0.8) -- (13.5 * 0.2 + 14.5 * 0.8, 3 * 0.2 + 5.5 * 0.8) -- (107);
	\end{pgfonlayer}
\end{tikzpicture}}
    \caption{\emph{A round in the Hydra game.} The curved edges represent potential intermediate nodes. The head Hercules targets is filled in black, its parent is marked \textcolor{red}{red}, and its grandparent \textcolor{blue}{blue}. The remaining subtree from the parent is shaded in \textcolor{red}{red}, and is duplicated at the end of the round. Here, the Hydra hasn't evolved, and hence only 3 new heads are grown. Observe that the subtree shaded \textcolor{blue-green}{blue green} is entirely unaffected.}
    \label{fig:hydra-game-representation}
\end{figure}

By fixing Hercules's strategy to any recursive function and considering the nondeterministic choices at Lines ~\ref{line:hydra-move1} and ~\ref{line:hydra-move2} (in \codeRef{fig:hydra-game})
demonically,
the progression of this game becomes the execution of a probabilistic program. 
The fact that this program has a \emph{finite expected runtime} for every possible nondeterministic scheduler (i.e., is $\PAST$) is observable from two facts: 
one, the deterministic hydra game only has a finite number of rounds \cite{KirbyParis1982}, and two: each round, in expectation, only takes a constant amount of time. 

Our goal in this section is to illustrate the apparatus required to prove that the game is $\PAST$.
Termination is usually demonstrated through \emph{ranking functions}.
In the original, deterministic, Hydra game, 
there is in fact a ranking function of the form discussed by \citet{Francez86} and \citet{manna1974mathematical} mapping program states to natural numbers tracking the upper bound on the remaining length of the game. 
This is because the supremum of the game's length from every state (varied over the strategies employed by Hercules) is always finite, in spite of the ordinals necessary to show this.
Unfortunately, because of nondeterministic choices, our variant does not have an upper bound on the expected runtime independent of the scheduler.

Prior work in proving $\BAST$ for probabilistic programs \cite{chatterjee2017termination, FioritiH15} uses
\emph{ranking supermartingales}, a generalization of ranking functions.
A ranking supermartingale maps program states to real values in such a way that in expectation, the function
strictly decreases by at least some minimum amount at each execution step.
Ranking supermartingales form a sound and complete proof rule for $\BAST$.
Unfortunately, we show that, despite a finite expected run time, we cannot find such a function for the stochastic Hydra game.
Indeed, we show that a termination argument for the Hydra game must use transfinite ordinals.



We begin by introducing, following \citet{KirbyParis1982}, a useful mapping $T$ from nodes in the Hydra to ordinals. 
The range of the mapping is $\varepsilon_0$, the smallest solution to the ordinal equation $x = \omega^x$.

\begin{definition}[Ordinal mapping of nodes in the Hydra]
    Let $\codeStyleMath{hydra} = (V, E)$ be a finite tree. Define the mapping $T : V \to \varepsilon_0$ with following properties:
    \begin{itemize}
        \item For every leaf node $v \in V$, $T(v) = 0$
        \item For every internal node $v \in V$ with children $v_1, v_2, \ldots v_m$ listed in decreasing order of the ordinals assigned to them by $T$,
        $$
        T(v) = \sum_{i = 1}^m \omega^{T(v_i)}
        $$
        In other words, $T(v)$ is the \emph{natural sum} of all $\omega^{T(v')}$ over all children $v'$ of the node $v$.
    \end{itemize}
\end{definition}

In each round of the Hydra game, if the Hydra survives, 
the ordinal assigned to the root of the Hydra by $T$ always reduces \cite{KirbyParis1982}. 
This is despite the increments to the regeneration capacity enabled by evolution.

In this work, we attempt to generalize ranking arguments to our setting.
We want to find ranking functions whose range are the ordinals such that they decrease in expectation in each step,
and only terminal states are given rank zero.
Since the ordinals are well-founded, this decrease in rank resembles the expected remaining length of execution. 
One could imagine that perhaps the ordinals are unnecessary and there is a clever encoding into existing ranking arguments, like the ones
by \citet{chatterjee2017termination}.
Unfortunately, we show that the naturals (or even the reals) cannot serve as an appropriate range for functions that 
guarantee an expected decrease of $1$ in each step.

Consider a starting point of a simple line Hydra of length 2 
that, after $n$ nondeterministic evolution steps, grows $4^n - 1$ new 
heads with probability $1/2^{n+1}$.
Note that the Hydra can no longer evolve or grow heads from this state, and
the game must hence be played for exactly $4^n + 1$ more steps to terminate.
%
%
%
Suppose there is a ranking supermartingale that assigns to the line Hydra a natural (or real) number $m$.
This function must necessarily assign to the new Hydra a value of greater than $4^n$.
Our requirements on the ranking function now imply that
$$
m \geq \frac{1}{2^{n + 1}} \times 4^n \implies m \geq 2^{n - 1}
$$
By engineering a sufficiently large value of $n$, the Hydra can invalidate this inequality. 
Hence, neither the naturals nor the reals can serve as a sufficient co-domain of the desired ranking function. 
However, the infinite ordinal $\omega$ is an excellent choice of rank for the line hydra.

Allowing ordinals in the ranking function creates a new challenge. 
What must be the rank of a state that, with some probability $0<p < 1$, can reach a state of rank $\omega$? 
For simplicity, we set the following additional requirement on our ranking functions: 
if, in one round, the game can reach a state with ordinal rank $x$ with positive probability $p$, 
then the source state must be ranked above $x$. 

With this additional property, we claim that \emph{the smallest appropriate ranking function for the Hydra game agrees with $T$ at all ordinal outputs.} 
This is because
$T$ assigns to the root the smallest ordinal greater than all ordinals reachable in a single step. 
We formalize this in \cref{lemma:ordinal-maintanence-property}.

\begin{lemma}
\label{lemma:ordinal-maintanence-property}
    From any Hydra $H$ with root node $r$ with $T(r) \geq \omega$, one can reach, in one step and with non-zero probabilities, an infinite sequence of hydras $H_1, H_2, \ldots$ with roots $r_1, r_2, \ldots$ such that the smallest ordinal larger than $T(r_1), T(r_2), \ldots$ is $T(r)$.
\end{lemma}

Hence, the smallest appropriate ranking function for our requirements is $T$, 
indicating that at the very least, all ordinals under $\varepsilon_0$ are needed to reason about the expected runtime 
of programs with both nondeterministic and probabilistic operators. 
In \cref{sec:proofs}, we see that ordinals up to the Church-Kleene ordinal $\CK$ are needed to reason about general probabilistic programs.
We include the proof of \cref{lemma:ordinal-maintanence-property} in \cref{sec:hydra-proof} for completeness.

In summary, our proof rule for proving $\PAST$ has three ingredients: a \emph{normal form} for programs in which
every probabilistic choice is of the form
${\small \codeStyleMath{skip}\ \oplus_{1/2} \codeStyleMath{exit} }$
(which, fortunately, the Hydra is already in),
an ordinal-valued ranking function (like the function $T$ above), and 
a proof that the rank decreases in an expected finite number of steps
despite nondeterminism (a $\BAST$ property, for which sound and complete proof rules exist).
Putting them together, we can argue that the stochastic Hydra is $\PAST$: the rank decreases in an expected finite number of steps
and the rank decreases a finite number of times until termination.
As we show later, the proof rule is semantically sound and complete, but the normal form is essential before it can be applied.

\section{Probabilistic Programs and their Termination}
\label{sec:prelims}

We now define the program model and the various notions of termination.

\subsection{Program Model}
\label{subsec:program-model}

The program model we employ is a straightforward nondeterministic extension of the language described by \citet{KaminskiKM19}. 
The syntax mirrors $\pGCL$, an extension of Dijkstra's Guarded Command Language ($\mathsf{GCL}$, \cite{dijkstra1976discipline}) that adds binary
probabilistic and nondeterministic choice operators.

\begin{definition}[Syntax of $\pGCL$]
    \label{def:syntax-pgcl}
    Let $\mathsf{Var}$ be a countable set of variable symbols. Programs in $\pGCL$ obey the grammar:
    $$
    \begin{aligned}
        \codeStyleMath{Prog} \Coloneqq \bot \mid v \coloneqq e \mid \codeStyleMath{Prog; Prog} \mid \codeStyleMath{Prog} \pChoice{p} \codeStyleMath{Prog} \mid \codeStyleMath{Prog} \nChoice \codeStyleMath{Prog} \mid \codeStyleMath{while}(b) \{\, \codeStyleMath{Prog};\, \}
    \end{aligned}
    $$
    where $v \in \mathsf{Var}$, $e$, $p$, and $b$ are arithmetical and boolean expressions over $\mathsf{Var}$, $\pChoice{p}$ is a \emph{probabilistic choice operator}, and $\nChoice$ is a \emph{nondeterministic choice operator}.
\end{definition}

$\bot$ here is the empty program. We omit the usual $\codeStyleMath{exit}$, $\codeStyleMath{skip}$, and $\codeStyleMath{if}$ structures for brevity, as they can easily be simulated in the mentioned syntax. Note the binary branching at probabilistic and nondeterministic operators.

In order to describe our semantics for $\pGCL$ programs, we need to formalize the notion of \emph{the scheduler}. Informally, a scheduler maps execution histories to actions at nondeterministic points in the program. Since the execution of $\pGCL$ programs can be uniquely determined from the sequence of decisions made at probabilistic and nondeterministic locations, we present the following more useful non-standard (but equivalent) definition for schedulers:
\begin{definition}[Scheduler]
    \label{def:scheduler}
    Let $\Sigma_n = \{L_n, R_n\}$ and $\Sigma_p = \{L_p, R_p\}$. A scheduler is simply a \emph{total} mapping from $(\Sigma_n \cup \Sigma_p)^* \to \Sigma_n$. Here, the alphabets $\Sigma_n$ and $\Sigma_p$ represent the \emph{Left} and \emph{Right} directions available at \emph{nondeterministic} and \emph{probabilistic} operators respectively.
\end{definition}

The following operational semantics for $\pGCL$ programs extends those of \citet{KaminskiKM19} with consideration for nondeterministic choice.
\begin{definition}[Semantics of $\pGCL$]
    \label{def:semantics-pgcl}
    Declare the following notations:
    \begin{itemize}
    \item $\setOfVariableValuations \triangleq \{ \eta \mid \eta : \setOfVariables \to \setOfRationals\}$ is the set of all possible \emph{variable valuations}.
    \item $\setOfPrograms$ is the collection of all programs derivable in the grammar specified in \cref{def:syntax-pgcl}.
    \item $\setOfSchedules$ is the set of all schedulers (defined in \cref{def:scheduler}).
    \item $\setOfExecutionStates \triangleq \setOfPrograms \times \setOfVariableValuations \times (\setOfPositiveRationals \cap [0, 1]) \times \{L_n, R_n, L_p, R_p\}^*$ is the set of all \emph{execution states}.
    \end{itemize}

    Additionally, let $\evaluationUnderContext{e}{\eta}$ and $\evaluationUnderContext{b}{\eta}$ be the evaluations of the arithmetical and boolean expressions $e$ and $b$ under the variable valuation $\eta \in \setOfVariableValuations$. The operational semantics of $\pGCL$ programs under a scheduler $f \in \setOfSchedules$ is defined by the smallest relation $\oneStepExec{f} \subseteq \setOfExecutionStates \times \setOfExecutionStates$ that complies with the inference rules illustrated in \cref{fig:semantics}.
    Furthermore, we define the transitive extensions $\nStepExec{f}{n}$ and $\starStepExec{f}$ by setting $\nStepExec{f}{1} \;\triangleq\; \oneStepExec{f}$ and for all $n \in \setOfNaturals$,
    $$
    (\sigma, \sigma') \in \;\nStepExec{f}{n+1} \Longleftrightarrow \exists \tau \in \setOfExecutionStates \cdot (\sigma, \tau) \in\; \nStepExec{f}{n} \land\; (\tau, \sigma') \in\; \oneStepExec{f}
\quad 
    \mbox{ and } \quad
    \starStepExec{f}\;\triangleq \bigcup_{i \in \setOfNaturals} \nStepExec{f}{i}
    $$

\begin{figure}[t]
\small
\[
\begin{array}[t]{@{}c@{}}
\inferrule[assign]{}
{
        (v \coloneqq e, \eta, a, w) \oneStepExec{f} (\bot, \eta[v \mapsto \evaluationUnderContext{e}{\eta}], a, w)
}\\[5mm]
\inferrule[concat1]{
        (P_1, \eta, a, w) \oneStepExec{f} (P'_1, \eta', a', w')
}{
        (P_1; P_2, \eta, a, w) \oneStepExec{f} (P'_1; P_2, \eta', a', w')
}\quad\quad\quad
\inferrule[concat2]{\ }{
        (\bot; P_2, \eta, a, w) \oneStepExec{f} (P_2, \eta, a, w)
}\\[5mm]
\inferrule[prob1]{
        \evaluationUnderContext{p}{\eta} \leq 0
}{
        (P_1 \pChoice{p} P_2, \eta, a, w) \oneStepExec{f} (P_2, \eta, a, w \cdot R_p)
}\quad\quad\quad
\inferrule[prob2]{
        \evaluationUnderContext{p}{\eta} \geq 1
}{
        (P_1 \pChoice{p} P_2, \eta, a, w) \oneStepExec{f} (P_1, \eta, a, w \cdot L_p)
}\\[5mm]
\inferrule[prob3]{
        0 < \evaluationUnderContext{p}{\eta} < 1
}{
        (P_1 \pChoice{p} P_2, \eta, a, w) \oneStepExec{f} (P_1, \eta, a \times \evaluationUnderContext{p}{\eta}, w \cdot L_p)
}\quad\quad\quad 
\inferrule[prob4]{
        0 < \evaluationUnderContext{p}{\eta} < 1
}{
        (P_1 \pChoice{p} P_2, \eta, a, w) \oneStepExec{f} (P_2, \eta, a \times (1 -\evaluationUnderContext{p}{\eta}), w \cdot R_p)
}\\[5mm]
\inferrule[nondet1]{
        f(w) = L_n
}{
        (P_1 \nChoice P_2, \eta, a, w) \oneStepExec{f} (P_1, \eta, a, w \cdot L_n)
}\quad\quad\quad
\inferrule[nondet2]{
        f(w) = R_n
}{
        (P_1 \nChoice P_2, \eta, a, w) \oneStepExec{f} (P_2, \eta, a, w \cdot R_n)
}\\[5mm]
\inferrule[loop1]{
        \evaluationUnderContext{b}{\eta} = 1
}{
        (\codeStyleMath{while}(b) \{ P; \} , \eta, a, w) \oneStepExec{f} (P; \codeStyleMath{while}(b) \{ P; \}, \eta, a, w)
}\quad\quad\quad
\inferrule[loop2]{
        \evaluationUnderContext{b}{\eta} = 0
}{
        (\codeStyleMath{while}(b) \{ P; \} , \eta, a, w) \oneStepExec{f} (\bot, \eta, a, w)
}
\end{array}
\]
\caption{Semantics of $\pGCL$}
\label{fig:semantics}
\end{figure}

\end{definition}
In \cref{def:semantics-pgcl}, $\setOfNaturals$ and $\setOfRationals$ are standard denotations for the sets of natural and rational numbers respectively. The operation $\eta[v \mapsto \evaluationUnderContext{e}{\eta}]$ is standard; it refers to the variable valuation that agrees with $\eta$ on all variables in $\setOfVariables \setminus \{v\}$ and assigns to $v$ the value $\evaluationUnderContext{e}{\eta}$. 
Like \citet{dijkstra1976discipline}, we restrict the range of values available to variables to the rationals; this avoids measure-theoretic apparatus that would be required otherwise \cite{BS78,TakisakaOUH21}.
Our semantics extends that of \citet{KaminskiKM19} in remembering the decisions made at nondeterministic execution states in addition to 
the branching at probabilistic states. 
This additional information facilitates compliance with the scheduler $f$.

To clarify later definitions, we distinguish the notion of the \emph{program state} from the execution state.
\begin{definition}[Program states]
    \label{def:program-states}
    A program state is simply a pair of $\pGCL$ program $P$ and a variable valuation $\eta$. The set of all program states $\setOfProgramStates$ is hence simply $\setOfPrograms \times \setOfVariableValuations$.
\end{definition}

For a fixed program $P \in \setOfPrograms$, we denote the \emph{initial program state} $(P, \eta_0)$ by $\initialProgramState{P}$, and the  \emph{initial execution state} $(P, \eta_0, 1, \varepsilon)$ by $\initialExecutionState{P}$. Here, $\varepsilon$ is the empty word in the language $(\Sigma_p \cup \Sigma_n)^*$ and unless otherwise specified, the initial variable valuation $\eta_0$ maps all variables in $\setOfVariables$ to $0$. Furthermore, every execution state of the form $(\bot, \eta, p, w)$ and program states of the form $(\bot, \eta)$ are said to be \emph{terminal}.
\begin{definition}[Execution tree]
    \label{def:execution-tree}
    Let $P$ be a $\pGCL$ program, $f$ be a scheduler, and let $\eta_0$ be some fixed initial variable valuation. Denote the \emph{initial execution state} $(P, \eta_0, 1, \varepsilon)$ by $\initialExecutionState{P}$. The execution tree of the program $P$ under the scheduler $f$ is the subgraph of $(\setOfExecutionStates, \oneStepExec{f})$ over the vertices
    $\{\sigma \in \setOfExecutionStates \mid \initialExecutionState{P} \starStepExec{f} \sigma \}$
\end{definition}
Our semantics ensures that the execution tree is a tree rooted at $\initialExecutionState{P}$. 
Incidentally, \cref{def:semantics-pgcl} agrees with the semantics defined in Definition 4 of \citet{KaminskiKM19}
for programs in $\setOfPrograms$ that do not use the nondeterministic choice operator $\nChoice$.

\subsection{Notions of termination}
\label{subsec:termination-definitions}

In this subsection, we formalize the various notions of termination motivated in \cref{sec:intro}. We begin with two necessary projection operations.
\begin{definition}
    $\executionStateProjectionProbability$ and $\executionStateProjectionHistory$ are total functions from the set of execution states $\setOfExecutionStates$ that satisfy
    $$
    \executionStateProjectionProbability((\_, \_, p, \_)) = p \qquad \text{and} \qquad \executionStateProjectionHistory((\_,\_,\_,w)) = w
    $$
    Here, $\_$ is shorthand for any arbitrary value.
\end{definition}

We now turn to termination probabilities.
Unlike deterministic programs, the probabilities of termination of $\pGCL$
programs depend on the scheduler employed to resolve non-determinism.
It is quite possible for a program to fully terminate under one
scheduler and run forever under another.

\begin{definition}[Termination Probability]
    \label{def:term-prob}
    Let $\terminalStateSet$ be a function that takes in a program state $\sigma = (P, \eta)$ and a scheduler $f$ and returns the set of \emph{terminal execution states} reachable from the corresponding initial execution state $\sigma_e = (P, \eta, 1, \varepsilon)$ under the scheduler $f$. Thus,
    $$
    \terminalStateSet(\sigma, f) \triangleq \{(\bot, \eta, p, w) \in \setOfExecutionStates \mid \sigma_e \starStepExec{f} (\bot, \eta, p, w) \}
    $$
    \emph{Termination probability} is a function that takes in a program state $\sigma$ and a scheduler $f$ and returns the probability of the termination of the execution initialized at $\sigma$ under $f$ by adding up the probabilities of states in $\terminalStateSet(\sigma, f)$:
    $$
    \termProb(\sigma, f) = \sum_{\sigma' \in \terminalStateSet(\sigma, f)} \executionStateProjectionProbability(\sigma')
    $$
\end{definition}

We now define the set $\AST$ that we motivated in \cref{sec:intro}.


\begin{definition}[Almost-sure termination]
    \label{def:nAST}
    $\AST$ (short for \emph{Almost-Surely Terminating}) is the set of all $\pGCL$ programs $P$ that yield a termination probability of $1$ from their initial states $\initialProgramState{P}$ under every possible scheduler $f \in \setOfSchedules$, i.e.,
    $$
    \AST = \{P \in \setOfPrograms \mid \sforall f \in \setOfSchedules \cdot \termProb(\initialProgramState{P}, f) = 1\}
    $$
\end{definition}
The symbol $\sforall$ indicates that $f$ is a second-order variable. This is necessary because the set $\setOfSchedules$ is not a countable entity. We will return to this detail in \cref{sec:proofs}.
  

Before we discuss the other notions of termination motivated in \cref{sec:intro}, we present
definitions for expected runtime.
We extend a useful presentation motivated by \citet{FioritiH15}: the
expected runtime is the sum of the infinite series of the probabilities of
surviving beyond $n$ steps.

\begin{definition}[Expected runtime]
    \label{def:exp-runtime}
    Let $\terminalStateSetInStep{\leq k}$ be a function that takes as input a program state $\sigma = (P, \eta)$ and a scheduler $f$ and returns the set of all terminal states reachable in $\leq k$ steps from the corresponding execution state $\sigma_e = (P, \eta, 1, \varepsilon)$ under the scheduler $f$:
    $$
    \terminalStateSetInStep{\leq k}(\sigma, f) = \{(\bot, \eta, p, w) \in \setOfExecutionStates \mid \exists n \in \setOfNaturals \cdot n \leq k \land \sigma_e \nStepExec{f}{n} (\bot, \eta, p, w) \}
    $$
    The \emph{expected runtime} from a program state $\sigma$ under the scheduler $f$ is the sum
    $$
    \expRuntime(\sigma, f) \triangleq \sum_{k \in \setOfNaturals} \left(1 - \sum_{\sigma' \in \terminalStateSetInStep{\leq k}(\sigma, f)}\executionStateProjectionProbability\left(\sigma'\right)\right)
    $$
\end{definition}
Observe that, as in the case of deterministic programs, the expected runtime can diverge.

We now present two notions: that of \emph{positive almost-sure termination} and \emph{bounded termination}. Positive almost-sure termination, introduced in \citet{BournezG05} and refined in \citet{FioritiH15}, describes programs that yield finite (meaning converging) expected runtimes under all schedulers. This finiteness property is captured by the existence of an upper bound on the series described in \cref{def:exp-runtime}.
\begin{definition}
    \label{def:nPAST}
    The set $\PAST$ contains precisely the $\pGCL$ programs that expect to terminate in a finite amount of time under any schedule, i.e.,
    $$
    \PAST \triangleq \left\{P \in \setOfPrograms \mid \sforall f \in \setOfSchedules \; \exists n \in \setOfNaturals \cdot \expRuntime(\initialProgramState{P}, f) < n \right\}
    $$
    As with $\AST$, the initial state $\initialProgramState{P}$ maps all variables in $\setOfVariables$ to $0$.
\end{definition}
The notion of bounded termination (introduced in
\citet{chatterjee2017termination}) is obtained by swapping the
positions of the quantifiers in \cref{def:nPAST}.

\begin{definition}
    \label{def:nbPAST}
    The set $\BAST$ contains precisely the $\pGCL$ programs that possess a finite upper bound over the expected runtimes across all schedules, i.e.,
    $$
    \nbPAST \triangleq \left\{P \in \setOfPrograms \mid \exists n \in \setOfNaturals \; \sforall f \in \setOfSchedules \cdot \expRuntime(\initialProgramState{P}, f) \leq n \right\}
    $$
\end{definition}

In the following sections, we study the decision problems $\AST$, $\PAST$, and $\BAST$, which ask: given a $\pGCL$ program $P$, is $P\in \AST$ (respectively $P \in \PAST$ and $P\in \BAST$)? Note that the variants of these problems without nondeterministic choice have already been explored by \citet{KaminskiKM19}.

\subsection{Recursion-theoretic preliminaries}

In order to precisely characterize the complexities of these decision problems, we need to introduce the \emph{arithmetical} and \emph{analytical} hierarchies of undecidability. Informally, these hierarchies describe increasingly undecidable problems by linking each problem to arithmetical formulas in first and second-order logic.
We only present relevant definitions here; for a full discussion of the properties of these hierarchies, see \citet{Rogers} and \citet{Kozen06}.

\begin{definition}[Arithmetical Hierarchy]
    \label{definition:arithmetical-hierarchy}
    Let $\calM_n$ be the set of all total Turing machines characterizing a subset of $\setOfNaturals^n$. For each natural number $n \geq 1$, the family of sets $\Sigma^0_n$ contains the set $L \subseteq \setOfNaturals$ \emph{iff} there exists a machine $M_L \in \calM_{n + 1}$ such that
    $$
    L = \left\{x \in \setOfNaturals \mid \exists y_1 \in \setOfNaturals\; \forall y_2 \in \setOfNaturals\; \cdots\, Q_n y_n \in \setOfNaturals \cdot M_L\left(x, y_1, \ldots y_n\right) = 1 \right\}
    $$
    where $Q_n$ is universal if $n$ is even and existential otherwise. Additionally, define $\Pi^0_n$ as the collection of sets $L \subseteq \setOfNaturals$ such that $(\setOfNaturals \setminus L) \in \Sigma^0_n$.

    The collections of sets $\left\{\Sigma^0_n\right\}$ and $\left\{\Pi^0_n\right\}$ form the \emph{Arithmetical Hierarchy} and any set $L \in \Sigma^0_n$ (or $L \in \Pi^0_n$) is said to be \emph{arithmetical}.
\end{definition}
\begin{definition}[Analytical Hierarchy]
    \label{definition:analytical-hierarchy}
    Let $\calM^m$ be the set of all total \emph{oracle} Turing machines with access to $m$ oracles, each characterizing a total function of the form $\setOfNaturals \to \setOfNaturals$. For each natural $n \geq 1$, call $\Sigma^1_n$ the collection of sets $L \subseteq \setOfNaturals$ with the property that each $L$ is associated with an $M_L \in \calM^n$ and
    $$
    L = \left\{x \in \setOfNaturals \mid \sexists f_1 \in \setOfNaturals \to \setOfNaturals \; \sforall f_2 \in \setOfNaturals \to \setOfNaturals \cdots\, \sQ_n f_n \in \setOfNaturals \to \setOfNaturals\; Q_{n+1} y \in \setOfNaturals \cdot M_L^{f_1, f_2, \ldots f_n}(x, y) = 1 \right\}
    $$
    Here, $M_L$ has oracle access to the functions $f_1, \ldots f_n$ and the quantifier $\sQ_n$ (and $Q_{n+1}$) is universal (resp. existential) if $n$ is even and existential (resp. universal) otherwise. The doubled symbols $\sforall$, $\sexists$, and $\sQ_n$ are \emph{second-order} quantifiers and the final quantifier $Q_{n + 1}$ is first-order. Let $\Pi^1_n$ be the collection of sets $L \subseteq \setOfNaturals$ with $(\setOfNaturals \setminus L) \in \Sigma^1_n$.

    The collections $\left\{\Sigma^1_n\right\}$ and $\left\{\Pi^1_n\right\}$ form the \emph{Analytical Hierarchy}. Any set $L \in \Sigma^1_n$ (or $L \in \Pi^1_n$) is said to be a (lightface) \emph{analytical} set. 
\end{definition}
The specific classes $\Sigma^1_1$ and $\Pi^1_1$ are referred to as the (lightface) \emph{analytic} and \emph{co-analytic} sets respectively.
It can be shown that both the Arithmetical and Analytical hierarchies are strict. Note that \cref{definition:analytical-hierarchy} details a \emph{normal form} for the Analytical hierarchy. In general, there can be arbitrarily many first-order variables after $\sQ_n$; sets defined in this way can always be redefined in the normal form \cite{Rogers}. Notice the implication that the \emph{first levels} of the analytical hierarchy (i.e., $\Sigma^1_1$ and $\Pi^1_1$ sets) contain every arithmetical set.

The strictness of these hierarchies motivates notions of completeness for these complexity classes.
\begin{definition}[Completeness]
    For any $\Gamma \in \bigcup_{n \in \setOfNaturals} \left\{\Sigma^1_n, \Pi^1_n, \Sigma^0_n, \Pi^0_n\right\}$, a set $L \subseteq \setOfNaturals$ is said to be \emph{$\Gamma$-hard} if, for every $L' \in \Gamma$, there exists a recursive procedure that maps $L'$ to $L$ and $\setOfNaturals \setminus L'$ to $\setOfNaturals \setminus L$. Furthermore, $L$ is \emph{$\Gamma$-complete} if it is $\Gamma$-hard and $L \in \Gamma$.
\end{definition}

\section{The Complexity of Probabilistic Termination}
\label{sec:proofs}

\citet{KaminskiKM19} showed that the decision problems $\AST$ and $\BAST$ are arithmetical in a language without nondeterministic choice.
Their proof can be extended to the setting with nondeterministic choice:

\begin{proposition}
    \label{prop:bast-ast}
  The decision problem $\AST$ is $\Pi^0_2$-complete and
  $\BAST$ is $\Sigma^0_2$-complete.
 \end{proposition}
We include the proof of \cref{prop:bast-ast} in \cref{sec:complexity-ast-bast} for completeness.
In contrast, we show that $\PAST$ is significantly harder.

\begin{theorem}
\label{th:PAST}
The decision problem $\PAST$ is $\Pi^1_1$-complete.
\end{theorem}

\subsubsection*{Upper Bound}
Expanding the series defining $\expRuntime$ in the definition of $\PAST$ gives
\begin{align}
    \PAST &= \left\{P \in \setOfPrograms \mid \sforall f \in \setOfSchedules \; \exists n \in \setOfNaturals \cdot \sum_{k \in \setOfNaturals} \left(1 - \sum_{\sigma \in \terminalStateSetInStep{\leq k}(\initialProgramState{P}, f)} \executionStateProjectionProbability\left(\sigma\right)\right) < n \right\} \nonumber \\
    \implies \PAST &= \left\{P \in \setOfPrograms \mid \sforall f \in \setOfSchedules \; \exists n \in \setOfNaturals \; \forall m \in \setOfNaturals \cdot \sum_{k \leq m} \left(1 - \sum_{\sigma \in \terminalStateSetInStep{\leq k}(\initialProgramState{P}, f)} \executionStateProjectionProbability\left(\sigma\right)\right) < n \right\} \label{eq:1}
\end{align}
It's quite easy to build a terminating program $M$ with oracle access to $f$ that, on inputs $m$ and $n$, computes the finite sum in the quantifier-free section of \cref{eq:1}.
This yields
    \begin{equation}
        \label{eq:past-spec}
    P \in \nPAST \Longleftrightarrow \sforall f \in \setOfSchedules\; \exists n \in \setOfNaturals\; \forall m \in \setOfNaturals \cdot M^f(P, n, m) = 1
    \end{equation}
\cref{eq:past-spec} is a characterization of $\nPAST$ that can be transformed into the normal form for $\Pi^1_1$ (as required by \cref{definition:analytical-hierarchy}) using equivalences detailed by \citet{Rogers}. Hence, $\nPAST \in \Pi^1_1$.

 \subsubsection*{Lower Bound: Recursion-Theoretic Preliminaries}

 To show the $\Pi^1_1$-hardness of $\nPAST$, we introduce a canonical $\Pi^1_1$-complete problem.
 Towards this, we define $\omega$-trees.
\begin{definition}[$\omega$-trees, well-founded $\omega$-trees, and recursive $\omega$-trees]
    Let $\setOfNaturals^*$ be the set of all finite sequences of natural numbers. Define the prefix relation $\prec_n\; \subseteq \setOfNaturals^* \times \setOfNaturals^*$ as
    $$
    w_1 \prec_n w_2 \Longleftrightarrow |w_1| < |w_2| \land \forall n \leq |w_1| \cdot w_1(n) = w_2(n)
    $$
    Here, $|w|$ stands for the length of the sequence and $w(n)$ refers to the $n^{th}$ element of $w$.

    The pair $(\setOfNaturals^*, \prec_n)$ is the \emph{complete $\omega$-tree}.
    An \emph{$\omega$-tree} is any subtree of the complete $\omega$-tree rooted at the empty sequence $\varepsilon$.
    An $\omega$-tree is \emph{well-founded} if there are no infinite branches in the tree.

    The characteristic function of an $\omega$-tree takes in sequences $w \in \setOfNaturals^*$ as input and returns $1$ when $w$ is a node in the tree and $0$ otherwise.
    An $\omega$-tree is \emph{recursive} if its characteristic function is decidable.
\end{definition}

Let $\setOfRecursiveWellFoundedTrees$ be the set of all total Turing machines that characterize \emph{well-founded recursive} $\omega$-trees. 
\begin{theorem}
    \label{theorem:well-founded-pi-1-1}
	$\setOfRecursiveWellFoundedTrees$ is $\Pi^1_1$-complete.
\end{theorem}

The proof of \cref{theorem:well-founded-pi-1-1} can be found in various textbooks \cite{Rogers,Kozen06}.
We will reduce $\setOfRecursiveWellFoundedTrees$ to $\nPAST$. 

\subsubsection*{Lower Bound: Reduction}
\label{subsec:past-reduction}
%
Our reduction leverages nondeterminism in selecting a branch in the complete $\omega$-tree. The remainder of the reduction traverses this branch in the input $\omega$-tree to check its finiteness.
\begin{figure}
\declareCodeFigure
\small
\begin{subfigure}{0.45\textwidth}
\begin{lstlisting}[
language=python, mathescape=true, escapechar=|, label={lst:npast-hardness-numgen}, xleftmargin=15pt]
def numGen():
    x, y, w $\coloneqq$ 0, 0, 0|\label{line:rand-gen-init}|
    while (y = 0):|\label{line:inner-loop-start}|
        x $\coloneqq$ x + 1
        y $\coloneqq$ 0 $\nChoice$ y $\coloneqq$ 1 |\label{line:nondeterministic-choice}|
        if (y = 1):
            break
        skip  $\pChoice{1/2}$ exit |\label{line:npast-gambler}|
        s $\coloneqq$ 2 * s |\label{line:inner-loop-end}|
    while (w < s):|\label{line:numgen-cheer-start}|
        w $\coloneqq$ w + 1 |\label{line:cheer}| |\label{line:cheer-entrance}|
    return x - 1|\label{line:numgen-end}|
\end{lstlisting}
\caption{Number generation procedure \codeStyleText{numGen}}
\label{fig:numgen-past-simulation}
\end{subfigure}
\begin{subfigure}{0.45\textwidth}
\begin{lstlisting}[
language=python, mathescape=true, escapechar=|, label={lst:npast-hardness-program}, xleftmargin=15pt,firstnumber=13]
node, s $\coloneqq$ [], 1 # Globals
while (True):
    x $\coloneqq$ numGen() |\label{line:npast-line-init}|
    node $\coloneqq$ node.append(x) |\label{line:npast-node-append}|
    z $\coloneqq$ execute(M, node) |\label{line:execute-m}|
    if (z = 0): |\label{line:simulation-edge-case-start}|
        n $\coloneqq$ numGen()|\label{line:simulation-pick-new-node-start}|
        while (n--):|\label{line:simulation-pick-node-real-start}|
            x $\coloneqq$ numGen()
            node $\coloneqq$ node.append(x)|\label{line:simulation-pick-new-node-end}|
        z $\coloneqq$ execute(M, node)|\label{line:simulation-check-new-node}|
        if (z = 1):
            infLoop()
\end{lstlisting}
\caption{The program $P_M$, calling \codeStyleText{numGen} several times.}
\label{fig:simulation}
\end{subfigure}
\caption{The reduction $P_M$ simulating the recursive $\omega$-tree $M$.}
\end{figure}

For every Turing machine $M$, we construct the $\pGCL$ program $P_M$. The program $P_M$ is detailed in \codeRef{fig:simulation}.
The simulation of $M$ by $P_M$, enabled by the Turing completeness of $\pGCL$ \cite{McIverM05}, is encapsulated by the function \codeStyleText{execute(M, node)}.
Here, the finite sequence of natural numbers \codeStyleText{node} is supplied to $M$ as input.

$P_M$ invokes a procedure called \codeStyleText{numGen} multiple times in its execution.
At a high level, \codeStyleText{numGen}, specified in \codeRef{fig:numgen-past-simulation}, makes use of nondeterminism to produce a distribution over natural numbers with the property that every ``successful'' execution of \codeStyleText{numGen} takes, in expectation, a roughly equal amount of time.
The first inner loop of \codeStyleText{numGen} (from Lines ~\ref{line:inner-loop-start} to ~\ref{line:inner-loop-end}) requires scheduler action at 
\listingLineRef{line:nondeterministic-choice} to safely exit. 
Notice that \codeStyleText{numGen} terminates execution with probability $1/2$ at every iteration of this loop; 
hence, the probability of staying inside the loop decreases exponentially the longer the loop is run. 
The variable \codeStyleText{x} tracks the number of iterations of the loop; the output of \codeStyleText{numGen} is \codeStyleText{x - 1}.
The global variable \codeStyleText{s} doubles each time the loop is run.
Being global, its value persists through multiple executions of \codeStyleText{numGen}.
The overall design of the reduction ensures that \codeStyleText{1/s} tracks the probability value $\executionStateProjectionProbability$ of the current non-terminal execution state.

The second loop (from Lines ~\ref{line:numgen-cheer-start} to ~\ref{line:cheer}) of \codeStyleText{numGen} induces its principal feature:
the stabilization of increments to the expected runtime of the reduction $P_M$ across all executions of \codeStyleText{numGen}.
It isn't difficult to show that the expected runtime increases by at least $1$ during each successful (i.e., reaching \listingLineRef{line:numgen-end} and returning a value) execution of \codeStyleText{numGen}.

\codeStyleText{numGen} is used by $P_M$ to pick a potential child of the current node in the recursive tree characterized by $M$ (stored by the variable \codeStyleText{node}); this is precisely why \codeStyleText{numGen} returns \codeStyleText{x - 1} at \listingLineRef{line:numgen-end}.
Accordingly, the output of \codeStyleText{numGen} is appended to the end of \codeStyleText{node} at \listingLineRef{line:npast-node-append}, and the presence of \codeStyleText{node} in the $\omega$-tree is then checked by $M$ at \listingLineRef{line:execute-m}. 
If \codeStyleText{node} is in the tree, \codeStyleText{execute(M, node)} returns $1$ at \listingLineRef{line:execute-m}, and the execution returns to \listingLineRef{line:npast-line-init} and picks another potential child of \codeStyleText{node}.
Observe that the mandatory singular call to \codeStyleText{numGen} in the child-choosing process increases the expected runtime of $P_M$ by at least $1$.
Consequently, if $M$ were to characterize an infinite branch, $P_M$ could explore this infinite branch and, in the process, make infinitely many calls to \codeStyleText{numGen}, pushing its expected runtime to infinity.

In a well-founded tree, $M$ will eventually return $0$ at \listingLineRef{line:execute-m}.
Suppose this happens at $\codeStyleMath{node}'$.
From then on out (i.e., from \listingLineRef{line:simulation-pick-new-node-start}), the program checks an edge case:
Lines ~\ref{line:simulation-pick-new-node-start} to ~\ref{line:simulation-pick-new-node-end} pick an arbitrary node of the full $\omega$-tree under $\codeStyleMath{node}'$, and the \codeStyleText{execute} call at \listingLineRef{line:simulation-check-new-node} checks if that node is in the tree validated by $M$.
If $M$ characterizes a tree (and not a graph), \codeStyleText{execute(M, node)} at \listingLineRef{line:simulation-check-new-node} will always return $0$.
Note that this check only requires, in expectation, a finite amount of additional time.

We now formally argue for the correctness of the intuitions provided above.

    \emph{Case 1: $M$ is not total.}
    This implies that there is some number $n$ for which $M$ does not halt.
    Accordingly, take the scheduler $f$ that, on the first execution of \codeStyleText{numGen}, exits its inner loop after $n + 1$ iterations.
    The input to $M$ at \listingLineRef{line:execute-m} is thus $n$.
    After reaching that line, $P_M$ runs indefinitely without ever altering its probability value.

    Suppose \listingLineRef{line:execute-m} is reached at the $m^{th}$ step with probability $p > 0$. For all $m' \geq m$, the probability of termination in $\leq m'$ steps must be bounded above by $1-p$.
    This is because the probability of non-termination at the $(m')^{th}$ step is $p$. Thus, the expected runtime $\expRuntime(\initialProgramState{P_M}, f)$ is
    \begin{equation*}
        \begin{split}
\sum_{k \in \setOfNaturals} \left(1 - \sum_{\sigma \in \terminalStateSetInStep{\leq k}(\initialProgramState{P}, f)} \executionStateProjectionProbability(\sigma) \right)
            \geq \sum_{m' \in \setOfNaturals^{\geq m}} \left(1 - \sum_{\sigma \in \terminalStateSetInStep{\leq m'}(\initialProgramState{P}, f)} \executionStateProjectionProbability(\sigma) \right)
            \geq \sum_{m' \in \setOfNaturals^{\geq m}} p
            = \infty
        \end{split}
    \end{equation*}
    proving this case.
  
    \begin{figure}[t]
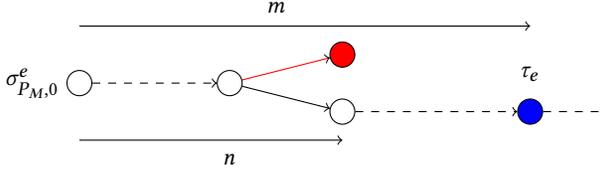

        \ctikzfig{mainproofs-gambler-tree}
        \caption{\emph{An execution from $\initialExecutionState{P_M}$}. The probabilistic operation at depth $n - 1$ yields one terminal and one non-terminal node at depth $n$. Hence, there is at most one non-terminal node at every depth. At depth $m$, the execution reaches $\tau_e$.}
        \label{fig:mainproofs-gambler-tree}
    \end{figure}
    \leavevmode
  \emph{Case 2: The scheduler chooses to never leave the first loop of \codeStyleText{numGen}.}
  We label these schedulers as \emph{badly behaved.}
  Let $f$ be one badly-behaved scheduler.
  The ``bad'' behaviour of $f$ can occur after many successful executions of \codeStyleText{numGen}.
  Suppose $P_M$ enters the first loop of \codeStyleText{numGen} at \listingLineRef{line:inner-loop-start} for the last time in its $m^{th}$ step with probability $p$.
  This implies an amassed termination probability of $(1-p)$ after $m$ steps.

  The design of \codeStyleText{numGen} (specifically, the available options at the probabilistic operation at \listingLineRef{line:npast-gambler}) indicates that in every execution tree rooted at the initial state $\initialExecutionState{P_M}$, there is at most one non-terminal execution state at every depth.
  Let the execution state at depth $m$ in the tree induced by the scheduler $f$ be $\tau_e$.
  Let $\tau$ be the program state corresponding to $\tau_e$ and $w_e = \executionStateProjectionHistory(\tau_e)$.
  Partitioning the expected runtime series $\expRuntime(\initialProgramState{P_M}, f)$ at the $m^{th}$ step gives
  \begin{flalign*}
            \expRuntime(\initialProgramState{P_M}, f)
            = \sum_{k \in \setOfNaturals^{< m}} \left(1 - \sum_{\sigma \in \terminalStateSetInStep{\leq k}(\initialProgramState{P_M}, f)} \executionStateProjectionProbability(\sigma) \right)
            + \sum_{k \in \setOfNaturals^{\geq m}} \left(1 - \sum_{\sigma \in \terminalStateSetInStep{\leq k}(\initialProgramState{P_M}, f)} \executionStateProjectionProbability(\sigma) \right)
  \end{flalign*}
    The series on the left is finite. The series on the right consists of the probabilities of the non-terminal execution states under $\tau_e$.

    Let $f'$ be the scheduler that satisfies $f'(u) = f(w_e u)$ for all histories $u \in (\Sigma_n \cup \Sigma_p)^*$. Let the execution tree from $\tau$ under $f'$ be $T'$ and the subtree of the execution tree from $\initialProgramState{P_M}$ under $f$ rooted at $\tau_e$ be $T$. Then, as far as the program states are concerned, $T$ and $T'$ are identical. This yields a natural mapping $g$ from nodes in $T$ to nodes in $T'$ with the property that $\executionStateProjectionProbability(\sigma) = p \times \executionStateProjectionProbability(g(\sigma))$ for every $\sigma \in T$.
    This means that the second series is just the expected runtime from $\tau$ under $f'$ scaled down by $p$:
    \begin{equation*}
            \sum_{k \in \setOfNaturals^{\geq m}} \left(1 - \sum_{\sigma \in \terminalStateSetInStep{\leq k}(\initialProgramState{P}, f)} \executionStateProjectionProbability(\sigma) \right) = p \times \expRuntime(\tau, f')
    \end{equation*}
    Because $f'$ never leaves the inner loop at \listingLineRef{line:inner-loop-start}, $\expRuntime(\tau, f')$ is finite; we omit the details for brevity.
    Thus, the expected runtime of $P_M$ under badly behaved schedulers is finite.

    A well-behaved scheduler is one that is not badly behaved.
    Well-behaved schedulers always exit \codeStyleText{numGen} with non-zero probability.
    Each well-behaved $f$ can be identified by the outputs that $f$ induces at executions of \codeStyleText{numGen}, and
    therefore every well-behaved scheduler corresponds to an infinite branch in the complete $\omega$-tree.
    From this point on, every machine $M$ is total and every scheduler $f$ is well-behaved.
    
    \emph{Case 3: $M$ fails to characterize a tree.} This means that the subgraph of the complete $\omega$-tree characterized by $M$ is disconnected. This indicates the existence of at least one broken branch, where $M$ returns $1$ until depth $m_1$, then returns $0$ until depth $m_1 + m_2$, and then returns $1$ again at depth $m_1 + m_2 + 1$, for some positive naturals $m_1$ and $m_2$.

    Let $f$ be the scheduler corresponding to this broken branch. Under $f$, $P_M$ will merrily execute onward until depth $m_1 + 1$, at which point \codeStyleText{execute(M, node)} at \listingLineRef{line:execute-m} will return $0$. This triggers the instructions under the \codeStyleText{if} condition at \listingLineRef{line:simulation-edge-case-start}, allowing $P_M$ to pick an arbitrary descendant of \codeStyleText{node}.

    Take the scheduler $f'$ that agrees with $f$ until depth $m_1 + 1$, returns $m_2$ at the \codeStyleText{numGen} call at \listingLineRef{line:simulation-pick-new-node-start}, and then picks the node in the broken branch at depth $m_1 + m_2 + 1$ included in the subgraph characterized by $M$ through the loop at Lines ~\ref{line:simulation-pick-node-real-start} to ~\ref{line:simulation-pick-new-node-end}. Under $f'$, $M$ will return $1$ at the \codeStyleText{execute(M, node)} call at \listingLineRef{line:simulation-check-new-node}, after which $P_M$ loops infinitely without ever altering its (positive) probability value. It's easy now to see that the expected runtime of $P_M$ under $f'$ is $+\infty$; we leave the details to the diligent reader.

    \emph{Case 4: $M$ characterizes a well-founded $\omega$-tree.} This means that every branch in the $\omega$-tree characterized by $M$ is finite. Every well-behaved $f$ thus begets a finite execution tree.
    Fix a well behaved $f$ and let $m$ be the depth of this finite tree. This means that $\terminalStateSetInStep{\leq m}(\initialProgramState{P_M}, f)$ is the set all leaves in the tree. Thus,
    $$
    \sum_{\sigma \in \terminalStateSetInStep{\leq m}(\initialProgramState{P_M}, f)} \executionStateProjectionProbability(\sigma) = 1 \implies \left(1 - \sum_{\sigma \in \terminalStateSetInStep{\leq m}(\initialProgramState{P_M}, f)} \executionStateProjectionProbability(\sigma)\right) = 0
    $$
    Since $m$ is the depth of the tree, for all $m' \geq m$, $\terminalStateSetInStep{\leq m'}(\initialProgramState{P}, f) = \terminalStateSetInStep{\leq m}(\initialProgramState{P}, f)$. These facts yield
    \begin{equation*}
        \begin{split}
            \expRuntime(\initialProgramState{P_M}, f) &= \sum_{k \in \setOfNaturals} \left(1 - \sum_{\sigma \in \terminalStateSetInStep{\leq k}(\initialProgramState{P_M}, f)} \executionStateProjectionProbability(\sigma) \right)
            = \sum_{k \leq m} \left(1 - \sum_{\sigma \in \terminalStateSetInStep{\leq k}(\initialProgramState{P_M}, f)} \executionStateProjectionProbability(\sigma) \right)
        \end{split}
    \end{equation*}
    This is a finite sum, meaning that the expected runtime is finite.

    \emph{Case 5: The $\omega$-tree characterized by $M$ has an infinite branch.} Let $f$ be the scheduler corresponding to this infinite branch. Observe that the execution tree of $P_M$ under $f$ must contain an infinite branch which calls \codeStyleText{numGen} infinitely often. Consequently, this branch enters the loop at \listingLineRef{line:cheer} infinitely often.

    Isolate one execution of this loop. Suppose the execution enters the loop with probability $p$ in its $m^{th}$ step. Then, the length of the loop is $\codeStyleMath{s} = 1/p$ and the execution exits the loop in its $(m+\codeStyleMath{s})^{th}$ step. Furthermore, the probability of non-termination at each step from $m$ to $(m + \codeStyleMath{s})$ is $p$. This means that
    $$
    \sum_{k = m}^{m+\codeStyleMath{s}} \left(1 - \sum_{\sigma \in \terminalStateSetInStep{\leq k}(\initialProgramState{P}, f)} \executionStateProjectionProbability(\sigma) \right) = p \times \codeStyleMath{s} = p \times (1/p) = 1
    $$
    Hence, the contribution to the expected runtime for $k \in \{m, m+1, \ldots m+ \codeStyleMath{s} \}$ is $1$. This result holds for all executions of the loop. Every execution of the loop thus corresponds to a constant increase to the expected runtime by $1$. Since the loop is executed infinitely often under $f$, the expected runtime under $f$ is $+\infty$.

    These five cases show that
    $$
    M \in \setOfRecursiveWellFoundedTrees \Longleftrightarrow P_M \in \nPAST
    $$
    Hence, $\nPAST$ is $\Pi^1_1$ hard.

\section{A proof rule for $\PAST$}
\label{sec:proof-rule-past}

The reduction proving the $\Pi^1_1$-hardness of $\PAST$ (detailed in \cref{sec:proofs}) uses the probabilistic choice operator $\pChoiceWithoutParamters$ in a very particular manner.
In effect, $\pChoiceWithoutParamters$ is only used to reduce the probability of continued execution.
This is realized by supplying $\pChoiceWithoutParamters$ with two options: one immediately terminating program execution and the other continuing it.
We use the term \emph{Knievel} to refer to these programs, reflecting the risky choices with terminal consequences made effortlessly by Evel Knievel.

\begin{definition}[\nameForNormalFormPrograms form for $\pGCL$ programs]
    \label{def:gambler-programs}
    A $\pGCL$ program $P$ is in \emph{\nameForNormalFormPrograms form} if every instance of the probabilistic choice operator in $P$ is of the form
    \begin{lstlisting}[language=python, numbers=none, mathescape=true, escapechar=|]
        skip |$\pChoice{p}$| exit
    \end{lstlisting}
    for any probability value $p$ and some fixed finite step implementation of the statements \codeStyleText{skip} and \codeStyleText{exit}.
\end{definition}

We now propose:
\begin{proposition}
    There is an effective transformation from any $\pGCL$ program $P$ into a program $P_K$ in \nameForNormalFormPrograms form such that 
    \begin{equation*}
        \qquad \qquad P \in \PAST \Longleftrightarrow P_K \in \PAST
    \end{equation*}
    \label{prop:effective-transform-to-past}
\end{proposition}

We now provide a brief sketch of the proof for \cref{prop:effective-transform-to-past}.
At a high level, the effective \nameForNormalFormPrograms form transformation involves two computable functions.
The first is induced by the $\Pi^1_1$-membership of $\PAST$ and the $\Pi^1_1$-completeness of $\setOfRecursiveWellFoundedTrees$. 
By definition of $\Pi^1_1$-hardness, there is a computable function, which we call $f$, 
that takes $\pGCL$ programs $P$ as input and outputs Turing machines $f(P)$ such that 
$P \in PAST$ \emph{iff} $f(P)$ characterizes a well-founded $\omega$-tree.
The second is the program schema we provided in \cref{sec:proofs} (more precisely, in \codeRef{fig:simulation}) to prove the $\Pi^1_1$-hardness of $\PAST$. 
More formally, it is a computable function $g$ that takes in Turing machines $M$ as input and produces $\pGCL$ programs $g(M)$ such that 
$M$ characterizes a well-founded $\omega$-tree \emph{iff} $g(M) \in \PAST$. 
Importantly, $g$ only outputs programs in \nameForNormalFormPrograms form.
The effective transformation is the composed function $g \circ f$, which satisfies the following properties: 
it is computable, its output is a Knievel program, and
for any $\pGCL$ program $P$, we have $P \in \PAST \Longleftrightarrow g(f(P)) \in \PAST$.

Note that we can derive direct constructions for $g \circ f$ from the reductions. We describe this in \cref{sec:gambler-transform}.

\begin{figure}[t]
    \begin{tikzpicture}
	\begin{pgfonlayer}{nodelayer}
		\node [style=plain circle] (1) at (-1, 0) {};
		\node [style=plain circle] (2) at (2, 0) {};
		\node [style=plain circle] (3) at (5, 0) {};
		\node [style=red circle] (5) at (4, 1.5) {};
		\node [style=red circle] (6) at (7, 1.5) {};
		\node [style=none] (7) at (10, 0) {};
	\end{pgfonlayer}
	\begin{pgfonlayer}{edgelayer}
		\draw [style=double directional] (1) to (2);
		\draw [style=double directional] (2) to (3);
		\draw [style=dashed directional edge] (3) to (7.center);
		\draw [style=red full edge] (2) to (5);
		\draw [style=red full edge] (3) to (6);
	\end{pgfonlayer}
\end{tikzpicture}
    \caption{\emph{An execution tree of a \nameForNormalFormPrograms form program.} The doubled lines represent the potential for multiple intermediate nodes. The mainline artery is depicted as the horizontal branch. Leaving this artery are single terminal nodes.}
    \label{fig:gambler-execution-tree}
\vspace*{-1mm}
\end{figure}
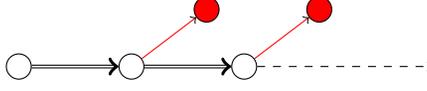
Therefore, the \nameForNormalFormPrograms form can be considered to be a kind of \emph{normal form} for $\PAST$ programs. Execution trees of these programs have a main arterial branch along which the \codeStyleText{skip} option is taken at every probabilistic operation. Branching away from this artery are leaves representing terminal states. See \cref{fig:gambler-execution-tree} for an illustration.

In this section, we present a proof rule for proving $\PAST$. 
We show that this proof rule is sound for programs in \nameForNormalFormPrograms form, and is complete for all $\PAST$ programs. 
Together with the effective transformation to \nameForNormalFormPrograms form,
our rule yields a semantically complete proof technique for $\PAST$. 

We begin with a few prerequisites.

\begin{definition}[Reachable States and Expected Time to Reach]
    \label{def:reachable-states}
    Let $\sigma = (P, \eta)$ be some program state, and $\sigma_e = (P, \eta, 1, \varepsilon)$ be its initial execution state. 
The set of states \emph{reachable} from $\sigma$ is 
    \begin{align*}
        \quad\setOfReachableStates{\sigma} \triangleq \left\{ (P', \eta') \in \setOfProgramStates \;\middle|\;
            \sexists f \in \setOfSchedules \; \exists p' \in \setOfPositiveRationals \; \exists w' \in (\Sigma_p \cup \Sigma_n)^* \cdot
            \sigma_e \starStepExec{f} (P', \eta', p', w')
        \right\}
    \end{align*}
    \label{def:expected-reachability-time}
Let $A \subseteq \setOfReachableStates{\sigma}$ be some subset of states reachable from $\sigma$. 
Call $A^k_{f, \sigma}$ the subset of execution states belonging to $A$ first reached in $k$ steps under the scheduler $f$. Formally,
    \begin{equation*}
        A^k_{f, \sigma} \triangleq \left\{ (P', \eta', p', w') \in \setOfExecutionStates \;\middle|\;
        \begin{aligned}
            &\exists p' \in \setOfPositiveRationals\; \exists w' \in (\Sigma_n \cup \Sigma_p)^* \cdot (P', \eta') \in A \land \sigma_e \nStepExec{f}{k} (P', \eta', p', w') \\
            &\land \left(\forall n < k \, \forall \tau \in A^n_{f, \sigma} \cdot \lnot (\tau \nStepExec{f}{k - n} (P', \eta', p', w')) \right)
        \end{aligned}
        \right\}
    \end{equation*}
    The second line ensures that there are no states belonging to $A$ along the path to the execution states in $A^k_{f, \sigma}$.
    
    Then, the expected time to reach $A$ from $\sigma$ under a scheduler $f$ is given by the series
    \begin{equation*}
            \expRuntimeToReach(\sigma, A, f) \triangleq \sum_{k \in \setOfNaturals} \Pr(\text{not reaching }A\text{ in }k\text{ steps}) 
            \triangleq \sum_{k \in \setOfNaturals} \left( 1 - \sum_{i = 1}^k  \sum_{\tau \in A^i_{f, \sigma}} \executionStateProjectionProbability(\tau)  \right)
    \end{equation*}
\end{definition}

In our proof rule, we use the notion of the \emph{Ranking Supermartingale Maps} (RSM-maps). 
RSM-maps have been proven to be a sound and complete proof technique for $\BAST$ by \citet{FuC19}. We mildly generalize their notions below.
\begin{definition}[RSM-maps]
    \label{def:rsm-map}
    Let $h : \setOfProgramStates \to \setOfReals$ be a function from the set of program states to the non-negative reals and $\epsilon > 0$ be an arbitrary real number. The pair $(h, \epsilon)$ is a \emph{Ranking Supermartingle Map} (RSM-map) \emph{iff} $h$ maps terminal states to $0$ and satisfies the following properties for every state $\sigma = (P, \eta)$ with $h(\sigma) > 0$:
    \begin{enumerate}
        \item For deterministic states $\sigma$ with their successors $\sigma' = (P', \eta')$ satisfying the property that
        $$
        \sforall f \in \setOfSchedules \cdot (P, \eta, 1, \varepsilon) \oneStepExec{f} (P', \eta', 1, \varepsilon)
        $$
        the function $h$ satisfies the following inequality:
        $$
        h(\sigma') + \epsilon \leq h(\sigma)
        $$
        
        \item For nondeterministic states $\sigma$ with successors $\sigma_l = (P_l, \eta_l)$ and $\sigma_r = (P_r, \eta_r)$ such that
        $$
        \sforall f \in \setOfSchedules \cdot (P, \eta, 1, \varepsilon) \oneStepExec{f} (P_l, \eta_l, 1, L_n) \lor (P, \eta, 1, \varepsilon) \oneStepExec{f} (P_r, \eta_r, 1, R_n)
        $$
        we have
        $$
        \max(h(\sigma_l), h(\sigma_r)) + \epsilon \leq h(\sigma)
        $$

        \item For probabilistic states $\sigma$ with the probability value $p$ and successors $\sigma_l = (P_l, \eta_l)$ and $\sigma_r = (P_r, \eta_r)$ such that
        $$
        \sforall f \in \setOfSchedules \cdot (P, \eta, 1, \varepsilon) \oneStepExec{f} (P_l, \eta_l, p, L_p) \land (P, \eta, 1, \varepsilon) \oneStepExec{f} (P_r, \eta_r, 1 - p, R_p)
        $$
        we have
        $$
        p \times h(\sigma_l) + (1 - p) \times h(\sigma_r) + \epsilon \leq h(\sigma)
        $$
    \end{enumerate}
    Note that every program state $\sigma$ is either deterministic, nondeterministic, probabilistic, or terminal.
\end{definition}
Unlike \citet{FuC19}, we do not require RSM-maps to only map terminal states to zero. 
Our goal is to use them to reason about the expected runtime to reach a collection of states. 
Towards this, we use the following two lemmas, showing soundness and completeness for $\BAST$, from \citet{FuC19}.
These are minor modifications of Lemmas~1 (Section 4.1) and 2 (Section 4.2) of \citet{FuC19}.

\begin{lemma}[Soundness of RSM-maps]
    \label{theorem:rsm-soundness}
    Let $(h, \epsilon)$ be an RSM-map. Denote by $\Sigma_{tgt}$ the set of states assigned $0$ by $h$, i.e.,
    $$
    \Sigma_{tgt} = \{ \sigma \in \setOfProgramStates \mid h(\sigma) = 0 \}
    $$
    Then, for all schedulers $f \in \setOfSchedules$ and states $\sigma \in \setOfProgramStates$, the expected runtime to reach $\Sigma_{tgt}$ is bounded above:
    $$
    \sforall f \in \setOfSchedules \; \forall \sigma \in \setOfProgramStates \cdot \expRuntimeToReach(\sigma, \Sigma_{tgt}, f) \leq \frac{h(\sigma)}{\epsilon}
    $$
\end{lemma}

\begin{lemma}[Completeness of RSM-maps]
    \label{theorem:rsm-completeness}
    Let $\sigma \in \setOfProgramStates$ be a program state, $\setOfReachableStates{\sigma}$ be the set of states reachable from $\sigma$, and $A_{tgt} \subseteq \setOfReachableStates{\sigma}$ be a target collection of states. Suppose that for all schedulers $f \in \setOfSchedules$, the expected runtime to reach $A_{tgt}$ from $\sigma$ bounded above by some $k \in \setOfReals$:
    $$
    \exists k \in \setOfReals\; \sforall f \in \setOfSchedules \cdot \expRuntimeToReach(\sigma, A_{tgt}, f) \leq k
    $$
    Then, there must exist an RSM-map $(h_\sigma, 1)$ such that $h_\sigma$ only assigns $0$ to states unreachable from $\sigma$ and to states in $A_{tgt}$, i.e.,
    $$
    h_\sigma(\tau) = 0 \Longleftrightarrow \tau \in A_{tgt} \cup \left( \setOfProgramStates \setminus \setOfReachableStates{\sigma} \right)
    $$
    Additionally, $h_\sigma(\sigma)$ upper bounds the expected runtime to reach $A_{tgt}$ under any scheduler. 
\end{lemma}

We now present our proof rule.
\begin{definition}[Proof rule for $\PAST$ programs in \nameForNormalFormPrograms form]
    \label{def:proof-rule}
    Let $\sigma_0 = (P, \eta_0)$ be an initial program state for the program $P \in \setOfPrograms$, $\ordinalVariableSymbol$ be some ordinal, and $\setOfReachableStates{\sigma_0}$ be the set of states reachable from $\initialProgramState{P}$. Let $g : \setOfReachableStates{\sigma_0} \to \ordinalVariableSymbol$ and $k : \setOfReachableStates{\sigma_0} \to \left(\left(\setOfProgramStates \to \setOfReals \right) \times \setOfReals \right)$ be functions that satisfy the following properties
    \begin{enumerate}
        \item For every $\sigma \in \setOfReachableStates{\sigma_0}$,
        $$
        g(\sigma) = 0 \Longleftrightarrow \sigma = (\bot, \_)
        $$
        In other words, $\sigma$ is terminal \emph{iff} $g(\sigma) = 0$.
    
        \item For a fixed \emph{non-terminal} state $\sigma \in \setOfReachableStates{\sigma_0}$, define the set $\setOfLowerStates{\sigma}$ as
        $$
        \setOfLowerStates{\sigma} \triangleq \{ \sigma' \in \setOfReachableStates{\sigma} \mid g(\sigma') < g(\sigma) \}
        $$

        The function $k$ returns an RSM-map $k(\sigma) = (h_\sigma, \epsilon_\sigma)$ that assigns $0$ \emph{only} to states \emph{not} reachable from $\sigma$ and to states in $\setOfLowerStates{\sigma}$, i.e.,
        $$
        \tau \in \left(\setOfLowerStates{\sigma} \cup \left(\setOfProgramStates \setminus \setOfReachableStates{\sigma}\right)\right) \Longleftrightarrow h_\sigma(\tau) = 0
        $$
    \end{enumerate}
    We refer to $g$ and $k$ as the \emph{rank} and \emph{certification} functions respectively.
\end{definition}
Notice that $\setOfLowerStates{\sigma}$ is simply the set of states reachable from $\sigma$ that are assigned a lower value by the rank $g$. Applying \cref{theorem:rsm-soundness}, we see that the RSM-Map $(h_\sigma, \epsilon_\sigma)$ certifies the fact that the expected time to reach $\setOfLowerStates{\sigma}$ from $\sigma$ under every scheduler $f$ is bounded above by a finite value.

\subsection{Partial Soundness}
\begin{figure}
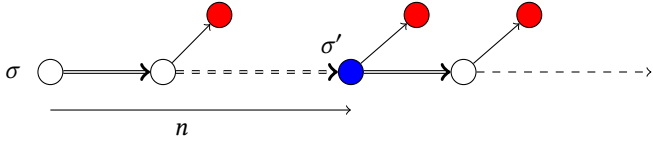

    \ctikzfig{soundness-reaching-lower}
    \caption{\emph{Case 2 in the proof of \cref{theorem:soundness-proof-rule}.} As earlier, the doubled lines indicate the potential for multiple intermediary nodes. The node $\sigma'$ is marked in \textcolor{blue}{blue}, and occurs after $n$ steps.}
    \label{fig:soundness-reaching-lower}
\end{figure}
We now show the soundness of this rule over \nameForNormalFormPrograms form programs.

\begin{theorem}
    \label{theorem:soundness-proof-rule}
    Let $P$ be a $\pGCL$ program in \nameForNormalFormPrograms form, $\sigma_0 = (P, \eta)$ be some initial program state, and $\ordinalVariableSymbol$ be some ordinal. Then, if two functions $g : \setOfReachableStates{\sigma_0} \to \ordinalVariableSymbol$ and $k : \setOfReachableStates{\sigma_0} \to \left(\left(\setOfProgramStates \to \setOfReals \right) \times \setOfReals \right)$ exist that satisfy the properties of the proof rule detailed in \cref{def:proof-rule}, then $P \in \PAST$.
\end{theorem}
\begin{proof}
    We show that from all states $\sigma \in \setOfReachableStates{\sigma_0}$ and all schedulers $f \in \setOfSchedules$, the expected runtime is finite.

    Let $S \subseteq \setOfReachableStates{\sigma_0}$ be the set of states from which the expected runtime is not finite. Assume the contrary and suppose $S \neq \emptyset$. Order states in $S$ by the values assigned to them by the rank $g$. Since $g$ assigns ordinals under $\ordinalVariableSymbol$, the well-ordering principle implies that $S$ has a least element. Denote this least element by $\sigma$.

    \emph{Case 1: $g(\sigma) = 0$.} This means $\sigma$ is a terminal state, immediately forming a contradiction.

    \emph{Case 2: $g(\sigma) = x$} for some $0 < x < \ordinalVariableSymbol$.
    Since $x > 0$, the soundness property of the RSM-map $k(\sigma)$ implies that $\setOfLowerStates{\sigma} \neq \emptyset$.
    Furthermore, the definitions of $S$ and $\sigma$ together imply that $\setOfLowerStates{\sigma} \cap S = \emptyset$. Therefore, for every scheduler $f \in \setOfSchedules$, the expected runtime from every $\sigma' \in \setOfLowerStates{\sigma}$ is finite.

    Take an arbitrary scheduler $f$. Denote the RSM-map $k(\sigma)$ by $(h_\sigma, \epsilon_\sigma)$. Combining the properties of $(h_\sigma, \epsilon_\sigma)$ detailed in \cref{def:proof-rule} and the results in \cref{theorem:rsm-soundness} gives us
    $$
    \expRuntimeToReach(\sigma, \setOfLowerStates{\sigma}, f) \leq \frac{h_\sigma(\sigma)}{\epsilon_\sigma}
    $$
    Take the execution tree corresponding to the scheduler $f$. There are two possibilities.

    \emph{Subcase 1: The tree never reaches any $\sigma' \in \setOfLowerStates{\sigma}$ in its main arterial branch.} In this case,
    $$
    \expRuntimeToReach(\sigma, \setOfLowerStates{\sigma}, f) = \expRuntime(\sigma, f)
    $$
    This is because the only states from $\setOfLowerStates{\sigma}$ in this tree are the terminal states leaving the arterial branch. The finiteness of the expected runtime follows immediately, forming a contradiction.

    \emph{Subcase 2: The tree reaches some $\sigma' \in \setOfLowerStates{\sigma}$ its main arterial branch for the first time after $n$ steps.} See \cref{fig:soundness-reaching-lower} for an illustration of this case. Call the probability value at the execution state corresponding to $\sigma'$ along the tree $p$.

    We now repeat an argument from \cref{subsec:past-reduction}.
    The expected runtime series from $\sigma$ under $f$ can be partitioned at the $n^{th}$ step, i.e., the point at which $\sigma'$ appears in the tree.
    The actions of the scheduler $f$ in the subtree rooted at $\sigma'$ must correspond to some scheduler $f'$ in an execution initialized at $\sigma'$.
    Furthermore, the membership of $\sigma' \in \setOfLowerStates{\sigma}$ implies that $\expRuntime(\sigma', f')$ is finite.
    The expected runtime series from $\sigma$ can now be written as
    \begin{equation*}
        \begin{split}
            \expRuntime(\sigma, f) &= \sum_{k \in \setOfNaturals} \Pr(\text{not terminating in }k\text{ steps}) \\
            &= \sum_{k \leq n} \Pr(\text{not terminating in }k\text{ steps}) + \sum_{k > n} \Pr(\text{not terminating in }k\text{ steps}) \\
            &\leq \frac{h_\sigma(\sigma)}{\epsilon_\sigma} + p \times \expRuntime(\sigma', f')
        \end{split}
    \end{equation*}
    This series is hence finite, forming a contradiction and completing the proof.
\end{proof}
\begin{remark}[Soundness for $\AST$]
    For every $\pGCL$ program $P$, if there exists rank and certification functions $f$ and $g$ that satisfy the properties laid out in \cref{def:proof-rule} from the initial state $\initialProgramState{P}$ of $P$, then $P$ is $\AST$. In other words, our proof rule is sound for $\AST$ over all $\pGCL$ programs, not just those in \nameForNormalFormPrograms form. We leave the proof for this to the reader; it's a simple extension of the proof of \cref{theorem:soundness-proof-rule}.
\end{remark}

\subsection{Total Completeness}
\label{subsec:completeness-rule}

We now discuss the completeness of the rule detailed in \cref{def:proof-rule}. 
Take an arbitrary (i.e., not necessarily Knievel) $\pGCL$ program $P$ and its initial state $\initialProgramState{P} = (P, \eta_0)$.

We describe a non-constructive procedure that yields candidates for the rank and certification functions $g$ and $k$. This procedure defines three unbounded sequences: one of partial ranks $\left\{ g_n \right\}$, one of partial certifications $\{ k_n \}$, and one of subsets of reachable states $\{ \Sigma_n \}$. Every partial rank $g_n$ maps a subset of values from $\setOfReachableStates{\initialProgramState{P}}$ to ordinals under the first non-recursive ordinal $\CK$. Each partial certificate $k_n$ returns RSM-maps for a subset of $\setOfReachableStates{\initialProgramState{P}}$, and each $\Sigma_n$ is a subset of $\setOfReachableStates{\initialProgramState{P}}$. Importantly, the lengths of these sequences can only be measured in ordinals. The domains over which the functions take values grow the further they get from the start.

In parallel to our construction, we will prove the following lemma.
\begin{lemma}
    \label{lemma:upper-bound-measure}
    For every ordinal $\ordinalVariableSymbol < \CK$, let $g_\ordinalVariableSymbol$ be the rank function used in \eqref{eq:measure-defn} at the end of the procedure specified in \cref{subsec:completeness-rule}. Then, for every state $\sigma \in \setOfProgramStates$, if $g_\ordinalVariableSymbol(\sigma)$ is defined, then $g_\ordinalVariableSymbol(\sigma) \leq \ordinalVariableSymbol$.
\end{lemma}

We begin by first defining $g_0$. Let $\Sigma_0 \subseteq \setOfReachableStates{\initialProgramState{P}}$ be the set of terminal states reachable from $\initialProgramState{P}$. For all $\sigma \in \Sigma_0$, set $g_0(\sigma)$ to $0$. Thus,
$$
\Sigma_0 = \left\{(\bot, \eta) \in \setOfReachableStates{\initialProgramState{P}} \right\} \qquad \text{and} \qquad \forall \sigma \in \Sigma_0 \cdot g_0(\sigma) = 0
$$
Since the proof rule only uses RSM-maps for non-terminal states, $k_0$ assigns to each $\sigma \in \Sigma_0$ an arbitrary RSM-map. Observe that \cref{lemma:upper-bound-measure} trivially holds for the base case of $g_0$.

We now describe a technique to derive the successor rank $g_{\ordinalVariableSymbol + 1}$ and certificate $k_{\ordinalVariableSymbol + 1}$ from $g_\ordinalVariableSymbol$ and $k_\ordinalVariableSymbol$ for every ordinal $\ordinalVariableSymbol$. We begin by requiring $g_{\ordinalVariableSymbol + 1}$ and $k_{\ordinalVariableSymbol + 1}$ to agree with $g_\ordinalVariableSymbol$ and $k_\ordinalVariableSymbol$ at every state they take values on:
$$
\forall \sigma \in \Sigma_\ordinalVariableSymbol \cdot g_{\ordinalVariableSymbol + 1}(\sigma) = g_\ordinalVariableSymbol(\sigma) \land k_{\ordinalVariableSymbol + 1}(\sigma) = k_\ordinalVariableSymbol(\sigma)
$$
We then define the set $\Sigma_{\ordinalVariableSymbol + 1}$:
$$
\Sigma_{\ordinalVariableSymbol + 1} \triangleq \{ \sigma \in \setOfReachableStates{\initialProgramState{P}} \mid \exists r \in \setOfReals \; \sforall f \in \setOfSchedules \cdot \expRuntimeToReach(\sigma, \Sigma_\ordinalVariableSymbol, f) \leq r \}
$$
Informally, $\Sigma_{\ordinalVariableSymbol+1}$ is the set of states from which the expected time to reach $\Sigma_\ordinalVariableSymbol$ is bounded above by a finite value. For each $\sigma \in \Sigma_{\ordinalVariableSymbol + 1}$, denote this bound by $r_\sigma$. Observe that, for $A_{tgt} = \Sigma_\ordinalVariableSymbol$, the bound on the expected time to reach $A_{tgt}$ satisfies the conditions outlined in \cref{theorem:rsm-completeness}. Hence, there must exist a RSM-map $(h_\sigma, 1)$ with
$$
h_\sigma(\tau) = 0 \Longleftrightarrow \tau \in \Sigma_\ordinalVariableSymbol \cup \left(\Sigma \setminus \setOfReachableStates{\sigma}\right)
$$
Simply set
$$
k_{\ordinalVariableSymbol+1}(\sigma) = (h_\sigma, 1)
$$
To determine the rank $g_{\ordinalVariableSymbol + 1}$ of a state $\sigma \in \Sigma_{\ordinalVariableSymbol + 1}$, we must analyze the subset of $\Sigma_\ordinalVariableSymbol$ reachable by an execution initialized at $\sigma$. Observe that, by the induction hypothesis of \cref{lemma:upper-bound-measure},
$$
\forall \sigma \in \Sigma_\ordinalVariableSymbol \cdot g_\ordinalVariableSymbol(\sigma) \leq \ordinalVariableSymbol
$$
Hence, the largest measure of any state in $\Sigma_\ordinalVariableSymbol$ is $\leq \ordinalVariableSymbol$. We can hence safely set, for all $\sigma \in \Sigma_{\ordinalVariableSymbol + 1}$,
$$
g_{\ordinalVariableSymbol + 1}(\sigma) = \ordinalVariableSymbol + 1
$$
This trivially satisfies the successor induction step in the formal proof of \cref{lemma:upper-bound-measure}.

%

We now detail the rank $g_\ordinalVariableSymbol$ and certificate $k_\ordinalVariableSymbol$ for any \emph{limit ordinal} $\ordinalVariableSymbol$. Be begin by defining
$$
\Sigma_\cup \triangleq \bigcup_{\ordinalVariableSymbol' < \ordinalVariableSymbol} \Sigma_{\ordinalVariableSymbol'}
$$
$\Sigma_\cup$ is thus the set of states that have been assigned a rank by some $g_{\ordinalVariableSymbol'}$. For every $\ordinalVariableSymbol' < \ordinalVariableSymbol$, set
$$
\forall \sigma' \in \Sigma_{\ordinalVariableSymbol'} \cdot g_\ordinalVariableSymbol(\sigma') = g_{\ordinalVariableSymbol'}(\sigma') \land k_\ordinalVariableSymbol(\sigma') = k_{\ordinalVariableSymbol'}(\sigma')
$$
This simply merges the domains of all functions defined for lower ordinals. Now, define $\Sigma_\ordinalVariableSymbol$ as
$$
\Sigma_\ordinalVariableSymbol \triangleq \{ \sigma \in \setOfReachableStates{\initialProgramState{P}} \mid \exists r \in \setOfReals \; \sforall f \in \setOfSchedules \cdot \expRuntimeToReach(\sigma, \Sigma_\cup, f) \leq r \}
$$
It's easy to see that $\Sigma_\ordinalVariableSymbol$ is the set of states from which the runtime for reaching the region of states ranked under $\ordinalVariableSymbol$ is bounded. For each $\sigma \in \Sigma_\ordinalVariableSymbol$, denote this bound by $r_\sigma$. Set $A_{tgt} = \Sigma_\cup$ and using $r_\sigma$, apply \cref{theorem:rsm-completeness} to derive the RSM-map $(h_\sigma, 1)$ for $A_{tgt}$ and set
$$
k_\ordinalVariableSymbol(\sigma) = (h_\sigma, 1)
$$
Similar to the previous case, $g_{\ordinalVariableSymbol'}(\sigma') \leq \ordinalVariableSymbol'$ for each state $\sigma' \in \Sigma_{\ordinalVariableSymbol'}$. Thus, for all $\sigma \in \Sigma_\ordinalVariableSymbol$, we set
$$
g_\ordinalVariableSymbol(\sigma) = \ordinalVariableSymbol
$$
Notice that this completes the proof of \cref{lemma:upper-bound-measure}.

Finally, we define the candidate rank $g : \setOfReachableStates{\initialProgramState{P}} \to \CK$ and certificate $k : \setOfReachableStates{\initialProgramState{P}} \to \mathbb{R}$ as
\begin{equation}
\tag{Candidate functions}
\label{eq:measure-defn}
g \triangleq \bigcup_{\ordinalVariableSymbol < \CK} g_\ordinalVariableSymbol \qquad\text{and}\qquad k \triangleq \bigcup_{\ordinalVariableSymbol < \CK} k_\ordinalVariableSymbol
\end{equation}
It's easy to see that, over the domains they're defined, $g$ and $k$ satisfy the requirements detailed in \cref{def:proof-rule}. We now prove that, for all $\PAST$ programs, $g$ and $k$ assign a value to the initial state $\initialProgramState{P}$. In this proof, we mildly abuse our notation and ascribe expected runtimes to execution trees; these are simply the expected runtimes to reach the leaves of the tree from the root of the tree under the scheduler that produces the tree.

\begin{lemma}
    \label{lemma:completeness}
    Let $P$ be a $\PAST$ program and let $g : \setOfReachableStates{\initialProgramState{P}} \to \CK$ and $k : \setOfReachableStates{\initialProgramState{P}} \to \left(\left(\setOfProgramStates \to \setOfReals \right) \times \setOfReals \right)$ be the candidate rank and certification functions defined in \eqref{eq:measure-defn}. Then, $g$ and $k$ are \emph{total}.
\end{lemma}
\begin{proof}
    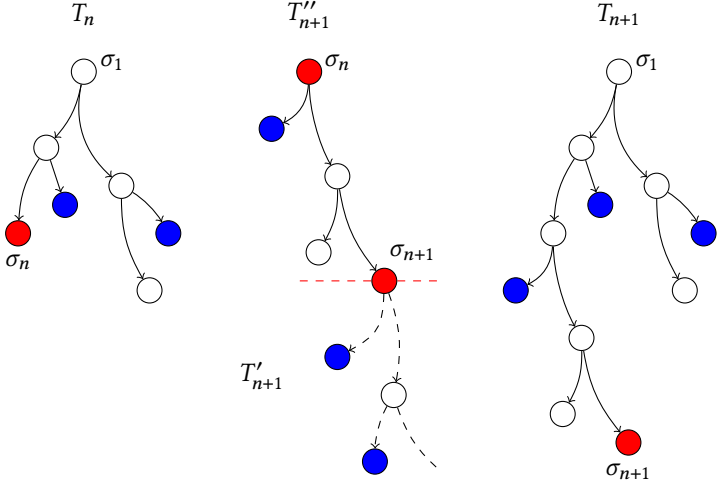
\begin{figure}
        \begin{tikzpicture}
	\begin{pgfonlayer}{nodelayer}
		\node [style=plain circle] (0) at (0, 4) {};
		\node [style=plain circle] (1) at (-1, 2) {};
		\node [style=plain circle] (2) at (1, 1) {};
		\node [style=plain circle] (3) at (1.75, -1.75) {};
		\node [style=blue circle] (5) at (2.25, -0.25) {};
		\node [style=red circle] (6) at (-1.75, -0.25) {};
		\node [style=blue circle] (7) at (-0.5, 0.5) {};
		\node [style=none] (8) at (0, 5.5) {$T_n$};
		\node [style=red circle] (9) at (6, 4) {};
		\node [style=blue circle] (10) at (5, 2.5) {};
		\node [style=plain circle] (11) at (6.75, 1.25) {};
		\node [style=plain circle] (12) at (6.25, -0.75) {};
		\node [style=red circle] (13) at (8, -1.5) {};
		\node [style=none] (14) at (6, 5.5) {$T''_{n+1}$};
		\node [style=plain circle] (15) at (14.25, 4) {};
		\node [style=plain circle] (16) at (13.25, 2) {};
		\node [style=plain circle] (17) at (15.25, 1) {};
		\node [style=plain circle] (18) at (16, -1.75) {};
		\node [style=blue circle] (19) at (16.5, -0.25) {};
		\node [style=blue circle] (21) at (13.75, 0.5) {};
		\node [style=none] (22) at (14.25, 5.5) {$T_{n+1}$};
		\node [style=plain circle] (23) at (12.5, -0.25) {};
		\node [style=blue circle] (24) at (11.5, -1.75) {};
		\node [style=plain circle] (25) at (13.25, -3) {};
		\node [style=plain circle] (26) at (12.75, -5) {};
		\node [style=red circle] (27) at (14.5, -5.75) {};
		\node [style=none] (28) at (-1.75, -1) {$\sigma_n$};
		\node [style=none] (29) at (0.75, 4.25) {$\sigma_1$};
		\node [style=none] (30) at (6.75, 4.25) {$\sigma_n$};
		\node [style=none] (31) at (8.75, -0.75) {$\sigma_{n + 1}$};
		\node [style=none] (32) at (14.5, -6.5) {$\sigma_{n + 1}$};
		\node [style=none] (33) at (15, 4.25) {$\sigma_1$};
		\node [style=blue circle] (34) at (6.75, -3.5) {};
		\node [style=plain circle] (35) at (8.25, -4.5) {};
		\node [style=blue circle] (36) at (7.75, -6.25) {};
		\node [style=none] (37) at (9.5, -6.5) {};
		\node [style=none] (38) at (5.75, -1.5) {};
		\node [style=none] (39) at (9.5, -1.5) {};
		\node [style=none] (40) at (4.75, -4) {$T'_{n + 1}$};
	\end{pgfonlayer}
	\begin{pgfonlayer}{edgelayer}
		\draw [style=direction edge, bend left=15] (0) to (1);
		\draw [style=direction edge, bend right] (0) to (2);
		\draw [style=direction edge, bend right=15] (1) to (6);
		\draw [style=direction edge, bend left=15, looseness=0.75] (2) to (5);
		\draw [style=direction edge, bend right=15] (2) to (3);
		\draw [style=direction edge] (1) to (7);
		\draw [style=direction edge, bend left] (9) to (10);
		\draw [style=direction edge, bend right=15, looseness=0.75] (9) to (11);
		\draw [style=direction edge, bend left=15] (11) to (12);
		\draw [style=direction edge, bend right=15] (11) to (13);
		\draw [style=direction edge, bend left=15] (15) to (16);
		\draw [style=direction edge, bend right] (15) to (17);
		\draw [style=direction edge, bend left=15, looseness=0.75] (17) to (19);
		\draw [style=direction edge, bend right=15] (17) to (18);
		\draw [style=direction edge] (16) to (21);
		\draw [style=direction edge, bend left] (23) to (24);
		\draw [style=direction edge, bend right=15, looseness=0.75] (23) to (25);
		\draw [style=direction edge, bend left=15] (25) to (26);
		\draw [style=direction edge, bend right=15] (25) to (27);
		\draw [style=direction edge, bend right=15] (16) to (23);
		\draw [style=dashed directional edge, bend left, looseness=1.25] (13) to (34);
		\draw [style=dashed directional edge, bend left=15] (13) to (35);
		\draw [style=dashed directional edge, bend right=15, looseness=0.75] (35) to (36);
		\draw [style=new edge style 0] (38.center) to (13);
		\draw [style=new edge style 0] (13) to (39.center);
		\draw [style=dashed simple edge, bend right=15] (35) to (37.center);
	\end{pgfonlayer}
\end{tikzpicture}
        \caption{\emph{The construction of $T_{n + 1}$.} In each of these trees, the \textcolor{blue}{blue} nodes represent states belonging to $\Sigma_{good}$. All other nodes belong to $\Sigma_{bad}$. The \textcolor{red}{red} nodes are bad leaf nodes selected for extension. The tree $T''_{n + 1}$ is produced by removing the subtree rooted at $\sigma_{n + 1}$ in $T'_{n + 1}$ (also depicted) without diminishing the expected time to reach $\Sigma_{good}$ too much. Observe that $T''_{n + 1}$ is merely attached to $\sigma_1$ to produce $T_{n + 1}$.}
        \label{fig:proofrule-completeness-construction}
    \end{figure}

    We prove this lemma by contradiction. Suppose $g$ and $k$ weren't total. It's easy to see that $g$ and $k$ are always defined over the same collection of states. Define
    $$
    \Sigma_{good} \triangleq \{ \sigma \in \setOfReachableStates{\initialProgramState{P}} \mid \exists \ordinalVariableSymbol < \CK \cdot g(\sigma) = \ordinalVariableSymbol \}
    $$
    In other words, $\Sigma_{good}$ is the collection of states reachable from the initial state $\initialProgramState{P}$ that are assigned a rank by $g$. Define
    $$
    \Sigma_{bad} \triangleq \setOfReachableStates{\initialProgramState{P}} \setminus \Sigma_{good}
    $$
    Our assumptions indicate that $\Sigma_{bad} \neq \emptyset$. They also indicate that all execution trees rooted at states in $\Sigma_{bad}$ yield finite expected runtimes. We claim that for these states, the expected time to reach $\Sigma_{good}$ is not bounded by a finite value.

    Why is this true? Take the set $A_{\sigma_b} \subseteq \Sigma_{good}$ of all good states reachable from some $\sigma_b \in \Sigma_{bad}$. Let $\ordinalVariableSymbol_b$ be the smallest ordinal larger than the ranks $g(\sigma_g)$ assigned to every good state $\sigma_g \in A_{\sigma_b}$. It isn't difficult to see that $\ordinalVariableSymbol_b$ is recursive, as $g(\sigma_g)$ is recursive for every good state $\sigma_g$ and there are countably many $\sigma_g \in A_{\sigma_b}$. Hence, $g_{\ordinalVariableSymbol_b} \subseteq g$ must be defined.
    If the expected time to reach $A_{\sigma_b}$ was bounded by some $r_{\sigma_b}$, the procedure forces
    $$
    g_{\ordinalVariableSymbol_b}(\sigma_b) \leq \ordinalVariableSymbol_b
    $$
    This forms a contradiction, justifying the inner claim.

    We now construct an infinite sequence of finite execution trees $\{ T_n \}$ rooted at some $\sigma_1 \in \Sigma_{bad}$ such that each $T_n$ has at least one bad state from $\Sigma_{bad}$ among its leaves and $T_{n+1}$ extends one of these leaves in $T_n$. Additionally, the expected runtime of each $T_n$ is bounded below by $n$. We then show that there exists a scheduler $f$ that produces the limit $T$ of $\{T_n\}$, and that the expected runtime from $\sigma_1$ under $f$ is $+\infty$.
    
    We begin with $T_1$. Take some state $\sigma_1 \in \Sigma_{bad}$. From our earlier arguments, we know that there must be some execution tree rooted at $\sigma_1$ that yields an expected time to reach $\Sigma_{good}$ at $r'_1$ steps with $r'_1 > 1$. Let $T'_1$ be this tree. The nature of the infinite series defined in \cref{def:expected-reachability-time} indicates that there must be a finite subtree of $T'_1$ that still yields a slightly lower expected reachability time $r_1$ with $1 \leq r_1 < r'_1$. Call this finite subtree $T_1$. Observe that there must be at least one bad state $\sigma_2 \in \Sigma_{bad}$ among the leaves of $T_1$; this arises from the strict inequality $r_1 < r'_1$. Furthermore, the expected runtime of $T_1$ is trivially above $1$.

    We now describe a procedure to build $T_{n + 1}$ from $T_n$. Take one bad leaf $\sigma_n \in T_n \cap \Sigma_{bad}$ reached with probability $p_n > 0$. We know that there must be an execution tree rooted at $\sigma_n$ with an expected time to reach $\Sigma_{good}$ of $r'_{n + 1} > \frac{1}{p_n}$. Call this tree $T'_{n + 1}$, and take the finite subtree $T''_{n + 1}$ of $T'_{n + 1}$ with an expected time to reach $\Sigma_{good}$ of $r_{n+1}$ with $\frac{1}{p} \leq r_{n + 1} < r'_{n+1}$. As before, the strict inequality means that there must be one bad leaf in $T''_{n + 1}$. Simply attach $T''_{n + 1}$ to the leaf $\sigma_n \in T_n$ to produce $T_{n + 1}$. This procedure is illustrated in \cref{fig:proofrule-completeness-construction}.

    Our construction guarantees that the expected runtime of $T_n$ is at least $n$. The construction of $T_{n + 1}$ implies that the expected runtime series of $T_{n + 1}$ simply extends that of $T_n$ with the probabilities of non-termination from $T''_{n+1}$. These new probabilities are weighted by $p_n$. Hence,
    \begin{equation*}
        \begin{split}
            \expRuntime(T_{n + 1}) &= \expRuntime(T_n) + p_n \times \expRuntime(T''_{n + 1})
            \geq n + p_n \times \frac{1}{p_n} = n + 1
        \end{split}
    \end{equation*}
    Hence, the expected runtime of $T_{n + 1}$ is at least $n + 1$, proving the primary property of the construction.

    Denote the limit of the sequence $\{T_n\}$ by $T$. Observe that the limit scheduler $f$ of the sequence of schedulers inducing each $T_n$ produces $T$ from $\sigma_1$. Furthermore, the expected runtime of $T$ must be infinite, as its subtrees $T_m \subset T$ ensure that it cannot bounded above by any $m \in \setOfNaturals$. This indicates that the program $P$ is not $\PAST$, forming a contradiction and completing the proof.
\end{proof}

We have thus shown
\begin{theorem}
    For each program $P \in \PAST$, there exist ranking and certification functions $g$ and $k$ that satisfy the requirements of the proof rule detailed in \cref{def:proof-rule}.
\end{theorem}

\subsection{All the Way to $\CK$}

\begin{figure}[t]
    \declareCodeFigure
    \small
    \begin{subfigure}{0.45\textwidth}
        \begin{lstlisting}[
%caption={$\INC$}, captionpos=b, 
language=python, mathescape=true, escapechar=|, label={lst:INC}, xleftmargin=15pt]
x, y $\coloneqq$ 1, 0
while (y = 0):|\label{line:INC-number-loop-start}|
    x $\coloneqq$ 2 * x
    y $\coloneqq$ 0 $\nChoice$ y $\coloneqq$ 1
    skip $\pChoice{1/2}$ exit|\label{line:INC-number-loop-end}|
while (x > 0):|\label{line:INC-cheer-loop}| 
    x $\coloneqq$ x - 1|\label{line:INC-cheer-loop-body}|
        \end{lstlisting}
        \caption{$\INC$}
        \label{fig:INC}
    \end{subfigure}
    \begin{subfigure}{0.45\textwidth}
        \begin{lstlisting}[
            % caption={$P_M$}, captionpos=b, 
            language=python, mathescape=true, escapechar=|, label={lst:P-M}, xleftmargin=15pt, firstnumber=8]
node $\coloneqq$ []
while (True):|\label{line:P-m-main-loop-start}|
    x, y $\coloneqq$ 0, 1
    while (y = 0):|\label{line:P-m-number-loop-start}|
        x $\coloneqq$ x + 1
        y $\coloneqq$ 0 $\nChoice$ y $\coloneqq$ 1
        skip $\pChoice{1/2}$ exit|\label{line:P-m-number-loop-end}|
    node $\coloneqq$ node.append(x)
    M_st $\coloneqq$ init_M(node)|\label{line:P-m-begin-M-operation}|
    while (not M_st.terminal()):
        M_st $\coloneqq$ M_step(M_state)
        skip $\pChoice{1/2}$ exit|\label{line:P-m-gamblers-choice}|
    if (M_st.reject()):|\label{line:P-m-near-end-of-loop}|
        exit|\label{line:P-m-end-M-operation}|
    execute($\INC$)|\label{line:P-m-INC-usage}|
            \end{lstlisting}
            \caption{The program $P_M$ using $\INC$}
            \label{fig:p_m}
    \end{subfigure}
    \caption{The full program $P_M$}
    \label{fig:full-p_m-for-ck}
\end{figure}

%
%
We now show, for every recursive ordinal $\ordinalVariableSymbol < \CK$, a $\PAST$ program in \nameForNormalFormPrograms form whose rank has range $\ordinalVariableSymbol$.
Together with the upper bound in the completeness argument, we conclude that $\CK$ is the appropriate range for the rank function $g$.

We begin with $\INC$ (see \codeRef{fig:INC}), a program for which the smallest rank that can be assigned to its initial state $\initialProgramState{\INC}$ is $2$. 
The execution of $\INC$ involves a scheduler-directed selection of a power of $2$ for the variable \codeStyleText{x} through the loop from Lines ~\ref{line:INC-number-loop-start} to ~\ref{line:INC-number-loop-end}.
After this selection is made, the program busy waits for \codeStyleText{x} many steps at Lines ~\ref{line:INC-cheer-loop} and ~\ref{line:INC-cheer-loop-body}.
The smallest rank that can be ascribed to states at \listingLineRef{line:INC-cheer-loop} is $1$, and
since \listingLineRef{line:INC-cheer-loop} can be reached in finitely many steps in expectation, 
the rank $2$ can be assigned to $\initialProgramState{\INC}$.
Furthermore, because $\INC \not\in \BAST$, a rank of $1$ cannot be ascribed to $\initialProgramState{\INC}$.

We now define programs for any recursive ordinal.
Lecture 40 of \citet{Kozen06} describes a mapping between well-founded recursive $\omega$-trees and recursive ordinals.
This involves the following finer mapping from the nodes of the $\omega$-trees to recursive ordinals:
all leaf nodes are assigned $0$ and all internal nodes are assigned the smallest ordinal larger than the values assigned to their immediate children. 
Finally, the tree is assigned the value of its root. 
Formally, for every recursive well-founded tree $M \in \setOfRecursiveWellFoundedTrees$, 
define a function $\ordMath_M : \setOfNaturals^* \to \CK$ as
\begin{equation*}
    \ordMath_M(w) = \begin{cases}
        0 & M(w) = 0 \lor \forall n \in \setOfNaturals \cdot M(\langle w, n\rangle) = 0\\
        \sup_{n \in \setOfNaturals} \ordMath_M(\langle w, n\rangle) + 1 & \text{otherwise}
    \end{cases}
\end{equation*}
The first line indicates that $\ordMath_M$ only maps leaves and nodes not validated by $M$ to $0$. 
Thus, every recursive ordinal $\ordinalVariableSymbol$ is associated with some $M \in \setOfRecursiveWellFoundedTrees$ such that $\ordinalVariableSymbol = \ordMath_M(\varepsilon)$.

For every $M \in \setOfRecursiveWellFoundedTrees$, we define a program $P_M$ (see \codeRef{fig:p_m}) that needs ordinals at least as large as $\ordMath_M(\varepsilon)$. 
As in \codeRef{fig:simulation}, $P_M$ nondeterministically traverses a branch in the tree identified by $M$. 
Each loop iteration begins with the choice of a candidate child \codeStyleText{x} through the inner loop at \listingLineRef{line:P-m-number-loop-start}.
The verification of the candidate child begins at \listingLineRef{line:P-m-begin-M-operation} and ends at \listingLineRef{line:P-m-end-M-operation}.
The functions \codeStyleText{init\_M} and \codeStyleText{M\_step} abstract the initialization and single-step execution of the machine $M$.
The structure \codeStyleText{M\_st} abstracts the current state of the execution of $M$ and provides options for checking whether that state is accepting or rejecting.
The insertion of Knievel's risk (continue or terminate) at \listingLineRef{line:P-m-gamblers-choice} inside the execution of $M$ (Lines ~\ref{line:P-m-begin-M-operation} to ~\ref{line:P-m-gamblers-choice}) constrains the expected runtime across all children against the running time of $M$. 
It isn't difficult to show that the expected runtime of each loop iteration from Lines ~\ref{line:P-m-main-loop-start} to ~\ref{line:P-m-near-end-of-loop} until the execution of $\INC$ at \listingLineRef{line:P-m-INC-usage} is bounded above by a small constant value. \emph{Call this constant value $r_M$.}

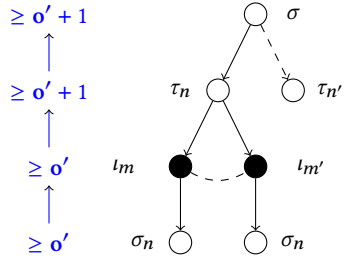
\begin{figure}
    \small
    \begin{tikzpicture}
	\begin{pgfonlayer}{nodelayer}
		\node [style=plain circle] (4) at (-2.5, -3) {};
		\node [style=plain circle] (7) at (-3.5, -5) {};
		\node [style=none] (8) at (-1.5, -3) {$\sigma$};
		\node [style=none] (14) at (-4.5, -5) {$\tau_n$};
		\node [style=black circle] (15) at (-4.5, -7) {};
		\node [style=black circle] (16) at (-2.5, -7) {};
		\node [style=none] (17) at (-6, -7) {$\iota_m$};
		\node [style=none] (18) at (-1, -7) {$\iota_{m'}$};
		\node [style=plain circle] (19) at (-4.5, -9) {};
		\node [style=plain circle] (20) at (-2.5, -9) {};
		\node [style=none] (21) at (-5.5, -9) {$\sigma_n$};
		\node [style=none] (22) at (-1.5, -9) {$\sigma_n$};
		\node [style=none] (23) at (-8, -9) {\textcolor{blue}{$\geq \mathbf{o}'$}};
		\node [style=none] (24) at (-8, -7) {\textcolor{blue}{$\geq \mathbf{o}'$}};
		\node [style=none] (25) at (-8, -5) {\textcolor{blue}{$\geq \mathbf{o}' + 1$}};
		\node [style=none] (26) at (-8, -3) {\textcolor{blue}{$\geq \mathbf{o}' + 1$}};
		\node [style=none] (27) at (-8, -7.5) {};
		\node [style=none] (28) at (-8, -8.5) {};
		\node [style=none] (29) at (-8, -6.5) {};
		\node [style=none] (30) at (-8, -5.5) {};
		\node [style=none] (31) at (-8, -3.5) {};
		\node [style=none] (32) at (-8, -4.5) {};
		\node [style=plain circle] (33) at (-1.5, -5) {};
		\node [style=none] (34) at (-0.5, -5) {$\tau_{n'}$};
	\end{pgfonlayer}
	\begin{pgfonlayer}{edgelayer}
		\draw [style=direction edge] (7) to (15);
		\draw [style=direction edge] (7) to (16);
		\draw [style=direction edge] (15) to (19);
		\draw [style=direction edge] (16) to (20);
		\draw [style=dashed simple edge, bend right] (15) to (16);
		\draw [style=blue directiona] (28.center) to (27.center);
		\draw [style=blue directiona] (29.center) to (30.center);
		\draw [style=direction edge] (4) to (7);
		\draw [style=blue directiona] (32.center) to (31.center);
		\draw [style=dashed directional edge] (4) to (33);
	\end{pgfonlayer}
\end{tikzpicture}
    \caption{\emph{The increment mechanism of $P_M$.} The branching in this execution tree is purely nondeterministic; the program can potentially reach any of its children. The minimum rank ascribable to the states is shown in \textcolor{blue}{blue}. $\sigma$ reaches $\tau_n$ in $r_M$ steps in expectation. From there, for each $m \in \setOfNaturals$, it reaches $\iota_m$, which can each be ascribed rank $\ordinalVariableSymbol'$. This causes the minimum possible rank value to increase.}
    \label{fig:ordinal-necessity-cascade}
\end{figure}
The proof for the $\PAST$ membership of $P_M$ is similar to the arguments contained in \cref{subsec:past-reduction}.
We do not repeat them here; instead, we discuss the executions of $P_M$ from program states beginning at the main loop (at \listingLineRef{line:P-m-main-loop-start}). 
These program states primarily differ in their values of \codeStyleText{node}, the `current' node in the tree recognized by $M$. 
They are consequently a natural link to the value of $\ordMath_M(\codeStyleMath{node})$.
We show:

\begin{lemma}
    \label{lemma:minimal-measure}
    Let $M \in \setOfRecursiveWellFoundedTrees$ be a well-founded recursive tree and $P_M$ be the program corresponding to it in \codeRef{fig:full-p_m-for-ck}. 
    Let $S_M$ be the set of program states at \listingLineRef{line:P-m-main-loop-start} of \codeRef{fig:p_m} reachable from the initial state $\initialProgramState{P_M}$. 
    Additionally, let $\node : S_M \to \setOfNaturals^*$ be a function that maps states in $S_M$ to the value of \codeStyleText{node} (i.e., the node) contained in them.

    Every rank that satisfies the rules detailed in the proof rule (\cref{def:proof-rule}) must assign to each $\sigma \in S_M$ an ordinal at least as large as $\ordMath_M(\node(\sigma))$.
\end{lemma}
\begin{proof}
    Observe that from every $\sigma \in S_M$, the execution begins with a scheduler-directed selection of a candidate child \codeStyleText{x} through the loop at Lines ~\ref{line:P-m-number-loop-start} - ~\ref{line:P-m-number-loop-end}.
    The expected runtime of $P_M$ under a scheduler that never picks a child, or picks a child not in the tree is trivially under $r_M$, the upper bound over the expected runtime of reaching \listingLineRef{line:P-m-near-end-of-loop} from \listingLineRef{line:P-m-main-loop-start}.
    The expected runtime under a scheduler that never exits the $\INC$ loop at Lines ~\ref{line:INC-number-loop-start} - ~\ref{line:INC-number-loop-end} is also similarly bounded.
    Hence, we only discuss schedulers picking actual children and actual values at $\INC$.

    We prove this lemma by transfinite induction on the value of $\ordMath_M(\node(\sigma))$.

    \emph{Base: $\ordMath_M(\node(\sigma)) = 0$.} This means $\node(\sigma)$ is a leaf. Therefore, the execution always reaches the terminal state at \listingLineRef{line:P-m-end-M-operation}, indicating that the expected runtime from $\sigma$ is bounded by $r_M$ under all schedulers. This justifies a rank assignment of $1$ to $\sigma$.

    \emph{Induction case 1: $\ordMath_M(\node(\sigma)) = \ordinalVariableSymbol + 1$} for some ordinal $\ordinalVariableSymbol$. This implies the existence of a child of $\node(\sigma)$ that was assigned the value $\ordinalVariableSymbol$ by $\ordMath_M$.

    Consider the selection of some child $n \in \setOfNaturals$ of $\node(\sigma)$ with $\ordMath_M(\langle\node(\sigma),n\rangle) = \ordinalVariableSymbol'$ and $\ordinalVariableSymbol' \leq \ordinalVariableSymbol$. 
    Call the program state in $S_M$ corresponding to this new node $\sigma_n$. By the induction hypothesis, the minimum rank that can be ascribed to $\sigma_n$ is $\ordinalVariableSymbol'$.

    From $\sigma$, the execution tree can select and validate the child $n$ within $r_M$ steps in expectation. After this, the execution enters $\INC$ and reaches \listingLineRef{line:INC-number-loop-start} of $\INC$; let $\tau_n$ be the program state at this stage. From $\tau_n$, the execution reaches \listingLineRef{line:INC-cheer-loop} of $\INC$ after selecting some $m \in \setOfNaturals$ for the variable \codeStyleText{x}. Call this program state $\iota_m$. From $\iota_m$, the execution reaches $\sigma_n$ in $m$ steps.

    We know, from the induction hypothesis, that $\sigma_n$ must be assigned a rank $\geq \ordinalVariableSymbol'$. This lower bound on the rank must also apply to $\iota_m$, as all executions from $\iota_m$ deterministically reach $\sigma_n$ in $m$ steps.
    However, from $\tau_n$, the execution can reach $\iota_m$ for any $m \in \setOfNaturals$. From each $\iota_m$, the expected runtime for reaching a lower ordinal is bounded below by $m$, an ever increasing quantity. Hence, the rank assigned to $\tau_n$ must at least be $\ordinalVariableSymbol' + 1$. Furthermore, because the execution can always expect to reach $\tau_m$ within $r_M$ steps, the state $\sigma$ can be assigned the same rank as $\tau_m$. See \cref{fig:ordinal-necessity-cascade} for an illustration. 
    
    Now, since there must be some $n \in \setOfNaturals$ such that $\sigma_n$ is ascribed $\ordinalVariableSymbol$, the state $\sigma$ must be ascribed a rank of at least $\ordinalVariableSymbol + 1$, completing this case.

    \emph{Induction case 2: $\ordMath_M(\node(\sigma)) = \ordinalVariableSymbol$} for some limit ordinal $\ordinalVariableSymbol$.
    This is only possible if there are countably many children under $\node(\sigma)$ and for every ordinal $\ordinalVariableSymbol' < \ordinalVariableSymbol$, there must be some child $n \in \setOfNaturals$ of $\node(\sigma)$ such that $\ordMath_M(\langle \node(\sigma), n \rangle) > \ordinalVariableSymbol'$.
    Let the program state in $S_M$ corresponding to the node $\langle \node(\sigma) , n\rangle$ be $\sigma_n$.
    Lifting the arguments from the previous case shows that the rank of $\sigma$ must be at least $\ordinalVariableSymbol' + 1$ for all $\ordinalVariableSymbol' < \ordinalVariableSymbol$. This forces the rank of $\sigma$ to be at least $\ordinalVariableSymbol$, completing this case, and therefore the proof.
\end{proof}
The initial program state of $P_M$ must thus be assigned a rank of at least $\ordMath_M(\varepsilon)$, justifying the need for ordinals up to $\CK$.

\subsection{\nameForNormalFormPrograms Form is Necessary}
\label{subsec:limitations-proof-rule}
While the rule defined in \cref{def:proof-rule} is complete for $\PAST$, it isn't sound for all programs. 
\begin{figure}
\begin{minipage}[c]{0.5\linewidth}
\begin{lstlisting}[
% caption={$\pGCL$ program that isn't $\PAST$ but satisfies the rule}, captionpos=b, 
language=python, mathescape=true, escapechar=|, label={lst:proof-rule-fail-example}, xleftmargin=15pt]
x, y $\coloneqq$ 0, 0
while (y = 0): |\label{lst:proof-rule-fail-while-loop-start}|
    x $\coloneqq$ x + 1
    y $\coloneqq$ 0 $\pChoice{1/2}$ y $\coloneqq$ 1
    if (y = 1): break
y $\coloneqq$ pow(4, x)
while (y > 0): y $\coloneqq$ y - 1
\end{lstlisting}
\end{minipage}
\quad\quad\quad\quad
\begin{minipage}[c]{0.3\linewidth}
    \begin{tikzpicture}
	\begin{pgfonlayer}{nodelayer}
		\node [style=plain circle] (0) at (0, 4) {};
		\node [style=plain circle] (1) at (-1, 2) {};
		\node [style=plain circle] (2) at (-2, 0) {};
		\node [style=none] (3) at (-3, -2) {};
		\node [style=plain circle] (4) at (1, 3) {};
		\node [style=plain circle] (5) at (0, 1) {};
		\node [style=plain circle] (6) at (-1, -1) {};
		\node [style=none] (7) at (-0.75, 4.25) {$\omega$};
		\node [style=none] (8) at (-1.75, 2.25) {$\omega$};
		\node [style=none] (9) at (-2.75, 0.25) {$\omega$};
		\node [style=none] (10) at (-0.25, -1.25) {$4^3$};
		\node [style=none] (11) at (0.75, 0.75) {$4^2$};
		\node [style=none] (12) at (1.75, 2.75) {$4$};
	\end{pgfonlayer}
	\begin{pgfonlayer}{edgelayer}
		\draw [style=direction edge] (0) to (4);
		\draw [style=direction edge] (0) to (1);
		\draw [style=direction edge] (1) to (5);
		\draw [style=direction edge] (1) to (2);
		\draw [style=direction edge] (2) to (6);
		\draw [style=dashed directional edge] (2) to (3.center);
	\end{pgfonlayer}
\end{tikzpicture}
\end{minipage}
    \caption{\emph{The unsoundness example.} (Left) a program that is not $\PAST$; (Right) An execution tree.
Each node on the leftmost branch is labelled $\omega$, and each node leaving that branch is labeled by $4^x$. It's easy to see that the expected runtime for this tree is $+\infty$.}
    \label{fig:unsoundness-example-proof-rule}
\end{figure}
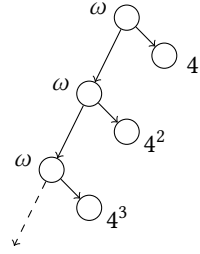
Take the program in \cref{fig:unsoundness-example-proof-rule}.
It is trivial to assign to all program states where the execution remains inside the first loop 
(at \listingLineRef{lst:proof-rule-fail-while-loop-start}) the rank $\omega$. 
We know that the expected runtime bound of $12$ (the expected runtime of the first loop) of exiting the loop yields some RSM-map for states inside the loop; 
simply assign to them this RSM-map. 
For all states leaving the loop, simply assign to them the value of $y$ and an RSM-map that sets $1$ to them and $0$ to everything else. 

This program is trivially not $\PAST$; however, the rank and certification functions we defined in the previous paragraph satisfy the properties of our proof rule. 
Thus, our rule must only be applied onto programs in \nameForNormalFormPrograms form to prove their membership in $\PAST$. 
Nevertheless, our total completeness argument indicates that if one could show that no valid rank and certification functions can exist for a particular $\pGCL$ program $P$, then $P \not\in \PAST$.

\section{Related Work}
\label{sec:related}

\subsubsection*{Termination and Fair Termination}
Termination is a classical problem in computer science, going back to Turing's paper \cite{Turing37}.
\emph{Ranking functions}, also known as \emph{progress measures}, are a standard technique for proving program termination. 
\citet{manna1974mathematical} described the use of such functions for demonstrating the termination of deterministic and nondeterministic programs. 
Their applicability for programs with unbounded nondeterminism has been explored \cite{Chandra78,HarelK84,Francez86}. 
The $\Pi^1_1$-completeness of the problem of determining if a program with these features halts is a result by \citet{Chandra78}, 
and the requirement for ordinals up to $\CK$ in these ranking functions was shown by \citet{AptP86}. 
Thus, there are recursive procedures transforming positively terminating probabilistic programs with bounded nondeterminism 
(i.e., $\PAST$ programs) to terminating non-probabilistic programs with unbounded nondeterminism, thereby ``compiling away'' the probabilities in the former.

\citet{Harel86} showed a general recursive tree transformation that reduced fair termination to termination in the setting of unbounded nondeterminism, thereby providing semantically sound and complete
proof rules for fair termination.
His reduction also proved the $\Pi_1^1$-completeness for fair termination.
We can study fairness in our context, and consider the natural $\nfAST$, $\nfPAST$, and $\nfBAST$ sets.
These quantify over the set of fair schedulers instead of the set of all schedulers.
For a general notion of strong fairness, we can show that $\nfAST$, $\nfPAST$, and $\nfBAST$ are all $\Pi_1^1$-hard and are in $\Pi^1_2$---the complexity gap
is due to a second, existential second-order quantifier over branches in an infinite tree
needed to capture fairness in the probabilistic setting.
When we restrict ourselves to the setting of \emph{finitary fairness} \cite{AlurH98}, which replaces the general fairness language with the largest safety language
contained within it, we see that $\nffAST$ and $\nffBAST$ remain $\Pi^0_2$ and $\Sigma^0_2$-complete,
and $\nffPAST$ remains $\Pi^1_1$-complete.
The appropriateness of finitary fairness for probabilistic programs have been argued before \cite{LengalLMR17}.
We include proofs for these statements in \cref{sec:discussion} for completeness.

\subsubsection*{Probabilistic Termination}
Termination for probabilistic programs is a well-studied area and trace their provenance to results on infinite-state Markov decision processes.
\emph{Ranking supermartingales} are regarded as the probabilistic generalization of ranking functions \cite{TakisakaOUH21}. 
Martingale based techniques have found applications in proving qualitative termination \cite{FuC19,HuangFC18,BournezG05,FioritiH15,ChakarovS13,AvanziniLY20}.
More recently, they have also been used in proving
quantitative termination, where one asks for the probability of termination 
\cite{ChatterjeeGMZ22,ChatterjeeNZ17,TakisakaOUH21,BeutnerO21}. 
Regarding these properties, the use of martingales in the determination of lower and upper bounds on the probability of termination has been shown by 
\citet{ChatterjeeGMZ22,ChatterjeeNZ17}. 
Futhermore, \citet{KuraUH19} have explored martingale-based approaches toward tail bounds on the expected runtime.

Our work is concerned with the qualitative properties of almost-sure and positive almost-sure termination. 
\citet{BournezG05} were the first to discuss the use of ranking supermartingales in a sound and complete proof technique for positive almost-sure termination of 
programs without nondeterminism. 
The extension of these rules for termination of programs with a global bound on the expected runtime across all schedulers (i.e., $\BAST$ programs) 
have been discussed by \citet{FioritiH15} and \citet{FuC19} with the former only including semi-completeness results and the latter proving completeness. 
Separately, sound and complete martingale-based proof rules for $\BAST$ (called \emph{strong} $\AST$ in their paper) have been explored by \citet{AvanziniLY20}.

Martingales have found applications in the study of almost-sure termination (i.e., $\AST$) as well. 
A sound proof rule for $\AST$ using martingales was described by \citet{ChakarovS13}, and \citet{McIverMKK18} paired supermartingales with 
certain intermediary progress functions in a widely applicable sound proof rule for almost-sure termination. 
Furthermore, algorithms for the synthesis of martingales for interesting subclasses of programs have been explored \cite{ChakarovS13,ChatterjeeFG16,ChatterjeeFNH18}.

Proof rules for $\AST$ and $\PAST$ that operate over the syntax of the programs have been studied \cite{McIverMKK18,KaminskiKMO18,OlmedoKKM16}. 
The most relevant are the rules that generate bounds on the expected runtime, presented by \citet{KaminskiKMO18}. 
Similar rules for recursive programs without loops have been presented by \citet{OlmedoKKM16}. 
Additionally, a relatively complete system with the ability to determine $\AST$ was introduced by \citet{BatzKKM21}. 
Importantly, none of these works include nondeterminism in their program models. Separately, algorithmic analyses of proof rules for $\AST$, $\PAST$, and non-termination have been discussed \cite{MoosbruggerBKK21}.
Interestingly, we do not know of a ``natural'' sound and complete proof rule for $\AST$.

Our focus on this paper is purely theoretical.
A number of papers have focused on automating the search for termination proofs by fixing a language for expressing ranking supermartingles
(e.g., linear or polynomial functions) and then using constraint solving to find appropriate functions \cite{ChakarovS13,ColonSS03,ChatterjeeFG16}.
We do not know of many algorithmic heuristics when ranks involve ordinals, even for non-probabilistic programs. 
Whether our proof rules can be automated in a sound way remains to be seen.
One could consider lexicographic ranking functions \cite{ChatterjeeGNZZ21,CookSZ13} as a first step, using the standard embedding of a tuple $(a_0, \ldots, a_n)$ to the ordinal sum
$a_0 \omega^n + \ldots + a_n$.

\subsubsection*{Complexity}
Finally, the complexities of $\AST$, $\PAST$, and other related decision problems for probabilistic programs with discrete distributions over their state spaces and 
without nondeterminism have been discussed in detail by \citet{KaminskiKM19}. 
Their results have been extended by \citet{BeutnerO21} to account for continuous distributions. 
As far as we know, the complexity analysis for nondeterministic extensions of these problems had not been studied before.
For $\AST$ and $\BAST$, the extensions are not difficult.
Our contribution is to notice the significantly higher complexity of $\PAST$.

\section{Conclusions}
\label{sec:conclusion}

We have characterized the complexity of $\PAST$ for $\pGCL$ programs with bounded nondeterministic and probabilistic choice operations.
We proved that this problem is $\Pi^1_1$-complete.
Using recursion-theoretic insights, we have defined an effectively computable normal form for $\pGCL$, 
and provided a sound and complete proof rules for $\PAST$ for normal form programs.
Our proof rule uses ordinals up to $\CK$ and this is necessary.
A specific implication of our results is that existing techniques based on ranking supermartinagles cannot be complete for $\PAST$.
%
%

\begin{acks}
We thank the reviewers for their helpful comments.
This research was sponsored in part by
the Deutsche Forschungsgemeinschaft project 389792660 TRR 248--CPEC
(see \url{https://perspicuous-computing.science}).
\end{acks}

\sloppy 


\label{beforebibliography}
\newoutputstream{pages}
\openoutputfile{main.pages.ctr}{pages}
\addtostream{pages}{\getpagerefnumber{beforebibliography}}
\closeoutputstream{pages}
\bibliography{bibliography}

\label{afterbibliography}
\newoutputstream{pagesbib}
\openoutputfile{main.pagesbib.ctr}{pagesbib}
\addtostream{pagesbib}{\getpagerefnumber{afterbibliography}}
\closeoutputstream{pagesbib}

\appendix

\section{Proof of Lemma \ref{lemma:ordinal-maintanence-property}} 
\label{sec:hydra-proof}

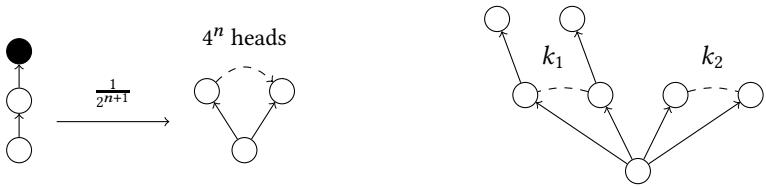
\begin{figure}[b]
    \begin{minipage}[c]{0.35\linewidth}
        \begin{tikzpicture}
	\begin{pgfonlayer}{nodelayer}
		\node [style=plain circle] (0) at (0, 0) {};
		\node [style=plain circle] (1) at (0, 1.3) {};
		\node [style=black circle] (2) at (0, 2.6) {};
		\node [style=plain circle] (4) at (6, 0) {};
		\node [style=plain circle] (7) at (7, 1.55) {};
		\node [style=plain circle] (8) at (5, 1.55) {};
		\node [style=none] (9) at (6, 3) {{\small $4^n$ heads}};
		\node [style=none] (10) at (1, 0.75) {};
		\node [style=none] (11) at (4, 0.75) {};
		\node [style=none] (12) at (2.5, 1.5) {{\footnotesize $\frac{1}{2^{n+1}}$}};
	\end{pgfonlayer}
	\begin{pgfonlayer}{edgelayer}
		\draw [style=direction edge] (0) to (1);
		\draw [style=direction edge] (1) to (2);
		\draw [style=direction edge] (4) to (7);
		\draw [style=direction edge] (4) to (8);
		\draw [style=dashed directional edge, bend left=45, looseness=1.25] (8) to (7);
		\draw [style=direction edge] (10.center) to (11.center);
	\end{pgfonlayer}
\end{tikzpicture}
    \end{minipage}
    \quad\quad\quad\quad
    \begin{minipage}[c]{0.3\linewidth}
        \begin{tikzpicture}
	\begin{pgfonlayer}{nodelayer}
		\node [style=plain circle] (0) at (0, 0) {};
		\node [style=plain circle] (1) at (1, 2) {};
		\node [style=plain circle] (2) at (3, 2) {};
		\node [style=plain circle] (3) at (-1, 2) {};
		\node [style=plain circle] (4) at (-3, 2) {};
		\node [style=plain circle] (5) at (-3.75, 4) {};
		\node [style=plain circle] (6) at (-1.75, 4) {};
		\node [style=none] (7) at (-2.25, 3) {$k_1$};
		\node [style=none] (8) at (2, 3) {$k_2$};
	\end{pgfonlayer}
	\begin{pgfonlayer}{edgelayer}
		\draw [style=direction edge] (0) to (4);
		\draw [style=direction edge] (0) to (3);
		\draw [style=direction edge] (0) to (1);
		\draw [style=direction edge] (0) to (2);
		\draw [style=direction edge] (3) to (6);
		\draw [style=direction edge] (4) to (5);
		\draw [style=dashed simple edge, bend left=345] (3) to (4);
		\draw [style=dashed simple edge, bend left=15] (1) to (2);
	\end{pgfonlayer}
\end{tikzpicture}
    \end{minipage}
    \caption{(Left) A move from the line hydra; (Right) A hydra with T value $k_1 \omega + k_2$.}
    \label{fig:simple-hydra}
\end{figure}
We restate the lemma here:
\begin{lemma}
    From any hydra $H$ with root node $r$ with $T(r) \geq \omega$, one can reach, in one step and with non-zero probabilities, an infinite sequence of hydras $H_1, H_2, \ldots$ with roots $r_1, r_2, \ldots$ such that the smallest ordinal larger than $T(r_1), T(r_2), \ldots$ is $T(r)$.
\end{lemma}

\begin{proof}
    We prove this by induction on $T(r)$.

    \emph{Base case: $r = \omega$.} This is precisely the case discussed in the left half of \cref{fig:simple-hydra}.
    
    From the hydra on the right, the game proceeds deterministically, with no potential for evolution. The $T$ value of the hydra on the right must therefore be $\geq 4^n$. Furthermore, it can be reached with probability $1/2^{n + 1}$ in one round from the Hydra on the left. This immediately implies that $r > 2^{n + 1} \times 4^{n - 1}$ for every $n \in \mathbb{N}$, which means that the smallest value for $r$ must be $\omega$. This proves this case.

    \emph{Induction step.} This time, $T(r) > \omega$. By the structure of the encoding, $T(r)$ is the natural sum of the $T$ values of its children. Let $x$ be the child assigned the largest ordinal among the root's children. There are two subcases now.

    \emph{Case 1: $T(x) < \omega$.} Then, $T(r) = a \omega + b$ for some naturals $a$ and $b$. This means that the structure of the Hydra is similar to the illustration in the right half of \cref{fig:simple-hydra}. To reiterate, there are $a$ branches of depth $2$ and $b$ branches of depth $1$.
    Since $a \neq 0$, take the game round where Hercules lops off a head at depth $2$. Observe that, for all $n \in \setOfNaturals$, if the Hydra chooses to evolve for exactly $n$ many times, the game has the probability of $1/2^{n+1}$ of reaching a state where the Hydra's root has a $T$ value of $(a - 1)\omega + b + 4^n$. If $b = 0$, this yields an infinite sequence of hydras as needed.
    
    If $b \neq 0$, there must be at least one leaf directly under the root. Simply take the case where Hercules removes one such leaf to produce a hydra with root at $a \omega + (b - 1)$. Augment this Hydra with the infinite sequence generated earlier to complete the proof.

    \emph{Case 2: $T(x) \geq \omega$.} In this case, simply apply the induction hypothesis on the subtree rooted at $x$. Let the sequence of ordinal values thus produced be $y_1, y_2, \ldots$. This yields an infinite sequence of Hydras $H_1, H_2, \ldots$ with $T$ values at the root differing only at the $\omega^{y_i}$ term for some $y_i \in {y_1, y_2, \ldots}$. It's trivial to see that the smallest ordinal larger than the $T$ values of the roots of $H_1, H_2, \ldots$ is $T(r)$, completing the lemma.
\end{proof}

\section{Arithmetic Complexity for $\AST$ and $\BAST$}
\label{sec:complexity-ast-bast}

\citet{KaminskiKM19} already prove the $\Pi^0_2$-hardness of $\nAST$ and the $\Sigma^0_2$-hardness of $\BAST$.
Intuitively, their proof effectively encodes canonical $\Pi^0_2$ and $\Sigma^0_2$-complete problems into $\AST$ and $\BAST$ respectively. Hence, we only need to show that $\nAST \in \Pi^0_2$ and $\BAST \in \Sigma^0_2$. 

We begin with $\AST$. Recall the definition of $\nAST$ from \cref{def:nAST}.
$$
\nAST = \left\{ P \in \setOfPrograms \mid \sforall f \in \setOfSchedules \cdot \sum_{\sigma \in \terminalStateSet(\initialProgramState{P}, f)} \executionStateProjectionProbability(\sigma) = 1 \right\}
$$
Generally, $\terminalStateSet(\initialProgramState{P}, f)$ can be an infinite set. Our semantics ensures that the probability values at execution states is never zero; hence, the $\AST$ series consists of strictly positive numbers. Its convergence to $1$ implies that for every rational $\delta < 1$, there must be a finite prefix of the series that exceeds $\delta$. If one were to order the elements of $\terminalStateSet(\initialProgramState{P}, f)$ by their distances from the initial execution state $\initialExecutionState{P}$, we get
$$
\nAST = \left\{ P \in \setOfPrograms \mid \forall f \in \setOfSchedules \; \forall \delta \in \setOfRationals^{(0, 1)} \; \exists k \in \setOfNaturals \cdot \sum_{\sigma \in \terminalStateSetInStep{\leq k}(\initialProgramState{P}, f)} \executionStateProjectionProbability(\sigma) > \delta \right\}
$$


To place $\nAST$ in $\Pi^0_2$, we need the notions of the \emph{partial schedule} and the \emph{scheduler tree}.
\begin{definition}[Partial schedule]
    A \emph{partial schedule} is any total function from the \emph{finite domain} $(\Sigma_n \cup \Sigma_p)^{\leq m}$ to $\Sigma_n$, for some $m \in \setOfNaturals$ where $(\Sigma_n \cup \Sigma_p)^{\leq m} = \cup_{k \leq m} (\Sigma_n \cup \Sigma_p)^k$. The \emph{size} of the partial schedule is $m$, the length of the longest word in the domain.
    
    Denote the \emph{set of all partial schedules of size $m$} by $\setOfPartialSchedulesOfSize{m}$, and the \emph{set of all partial schedules} by $\setOfPartialSchedules$. The \emph{standard extension} of the partial schedule $f_m$ of size $m$ is the scheduler $f$ such that
    $$
    f(w) = \begin{cases}
        f_m(w) & |w| \leq m\\
        L_n & |w| > m
    \end{cases}
    $$
\end{definition}
The structure of partial schedules induces a natural ancestry relation that yields the scheduler tree.
\begin{definition}[Scheduler Tree]
    Define the relation $\prec\; \subset \setOfPartialSchedules \times \setOfPartialSchedules$ as follows.
    $$
    f_m \prec f_n \Longleftrightarrow m < n \land \forall w \in (\Sigma_n \cup \Sigma_p)^{\leq m} \cdot f_m(w) = f_n(w)
    $$
    where the sizes of $f_m$ and $f_n$ are $m$ and $n$ respectively.

    The pair $(\setOfPartialSchedules, \prec)$ forms the \emph{scheduler tree}. To complete the tree, we denote by $f_0$ the \emph{empty scheduler}, and set $f_0 \prec f_n$ for all $n > 0$.
\end{definition}
Clearly, $f_0$ is the root of the scheduler tree. Observe that all partial schedules of size $m$ are present at depth $m$. Furthermore, every infinite branch in the scheduler tree corresponds to a full scheduler, and every full scheduler can be associated to a single infinite branch in the tree. We leave the details of this bijection to the diligent reader. Note that the binary nature of the branching at nondeterministic and probabilistic operations means that the scheduler tree is \emph{finitely branching}.

We now show certain useful properties of $\nAST$ programs.
\begin{lemma}
    \label{lemma:finitely-many-marked-nodes} 
    Fix a $\nAST$ program $P$. For every rational $0 < \delta < 1$ and schedule $f$, call the smallest $k$ that satisfies the inequality
    $$
    \sum_{\sigma \in \terminalStateSetInStep{\leq k}(\initialProgramState{P}, f)} \executionStateProjectionProbability(\sigma) > \delta
    $$
    the \emph{required simulation time} to cross $\delta$ for the schedule $f$. Then, the set of all required simulation times of the program $P$ to cross the threshold $\delta$ under all possible schedules $f$ has an upper bound.
\end{lemma}
\begin{proof}
    Since $P \in \nAST$, every scheduler $f$ is associated with some required simulation time $k_f$. For each $f$, isolate the infinite branch corresponding to $f$ in the scheduler tree and \emph{mark} the $k_f^{th}$ node along that branch. See \cref{fig:marked-nodes-scheduler-tree} for an illustration. We will prove that there can only be finitely many marked nodes in the tree.

    We first show that there is exactly one marked node in every branch. We do so by deriving a contradiction after assuming the contrary. Take some branch with two marked nodes at distances $k_1$ and $k_2$ with $k_1 < k_2$. Let the schedulers corresponding to the marked nodes be $f_1$ and $f_2$. Since they share a prefix in the tree, $f_1$ and $f_2$ \emph{must} agree on all $(\Sigma_n \cup \Sigma_P)^{\leq k_1}$. Moreover, since $k_1$ is the required simulation time for $f_1$, $\terminalStateSetInStep{\leq k_1}(\initialProgramState{P}, f_1)$ contains enough program states to amass a probability of termination greater than $\delta$.
    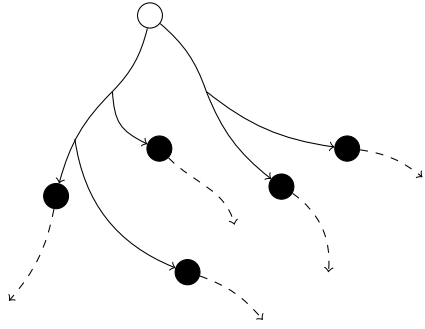
\begin{figure}
        \begin{tikzpicture}
	\begin{pgfonlayer}{nodelayer}
		\node [style=plain circle] (0) at (0, 2) {};
		\node [style=black circle] (2) at (-2.5, -2.75) {};
		\node [style=black circle] (3) at (0.25, -1.5) {};
		\node [style=black circle] (4) at (1, -4.75) {};
		\node [style=black circle] (5) at (3.5, -2.5) {};
		\node [style=none] (6) at (-1, 0) {};
		\node [style=none] (7) at (-2, -1.25) {};
		\node [style=none] (8) at (1.5, 0) {};
		\node [style=black circle] (9) at (5.25, -1.5) {};
		\node [style=none] (10) at (-3.75, -5.5) {};
		\node [style=none] (11) at (3, -6) {};
		\node [style=none] (12) at (2.25, -3.5) {};
		\node [style=none] (13) at (4.75, -4.75) {};
		\node [style=none] (14) at (7.25, -2.25) {};
	\end{pgfonlayer}
	\begin{pgfonlayer}{edgelayer}
		\draw [bend left=15] (0) to (6.center);
		\draw [style=direction edge, bend right=15] (6.center) to (2);
		\draw [style=direction edge, bend right, looseness=1.25] (6.center) to (3);
		\draw [style=direction edge, bend right] (7.center) to (4);
		\draw [bend left=15] (0) to (8.center);
		\draw [style=direction edge, bend right=15] (8.center) to (5);
		\draw [style=direction edge, bend right=15] (8.center) to (9);
		\draw [style=dashed directional edge, bend left=15] (2) to (10.center);
		\draw [style=dashed directional edge, bend left=15] (4) to (11.center);
		\draw [style=dashed directional edge, in=105, out=-45] (3) to (12.center);
		\draw [style=dashed directional edge, bend left] (5) to (13.center);
		\draw [style=dashed directional edge, bend left=15] (9) to (14.center);
	\end{pgfonlayer}
\end{tikzpicture}
        \caption{Marked nodes (filled in black) in the scheduler tree. Notice that the tree is finitely branching, and that no marked node is an ancestor of another.}
        \label{fig:marked-nodes-scheduler-tree}
    \end{figure}

    These two facts make it apparent that $\terminalStateSetInStep{\leq k_1}(\initialProgramState{P}, f_1) = \terminalStateSetInStep{\leq k_1}(\initialProgramState{P}, f_2)$. Additionally, since $k_1 < k_2$, $\terminalStateSetInStep{\leq k_1}(\initialProgramState{P}, f_2) \subsetneq \terminalStateSetInStep{\leq k_2}(\initialProgramState{P}, f_2)$. This means that $k_1$ is a smaller simulation time for $f_2$, contradicting the minimality of $k_2$ for $f_2$.

    We now know that there can only be a single marked node in each branch in the scheduler tree. Suppose there are infinitely many marked nodes. These marked nodes must be spread out over infinitely many branches. Form a subtree of the scheduler tree by lopping off the children of the marked nodes. Consequently, all marked nodes in the newly formed subtree are leaves. Our assumption indicates that there are infinitely many leaves in this tree. But, the tree is finitely branching. König's lemma indicates the presence of an infinite branch in this tree, which indicates the presence of a branch \emph{without} marked nodes (as only leaves are marked)! This contradicts the $\nAST$ nature of $P$.

    Hence, there are only finitely many marked nodes. This means that the set of all required simulation times is finite, which trivially indicates the presence of an upper bound, proving the lemma.
\end{proof}

We derive the $\Pi^0_2$ formula for $\nAST$ using \cref{lemma:finitely-many-marked-nodes}.

\begin{theorem}
    $\nAST \in \Pi^0_2$.
\end{theorem}
\begin{proof}
    Consider a turing machine $M$ that takes in three inputs: a rational $\delta$ between $0$ and $1$, a natural number $n$, and a $\pGCL$ program $P$. The outer loop of $M$ traverses the partial schedules at depth $n$ in the scheduler tree. The finitely branching nature of the scheduler tree indicates a finite number of partial schedules at this depth. Once $M$ selects a partial schedule $f_n$, it produces the execution tree of $P$ under $f$ till depth $n$. It does so by simulating $P$ for all possible probabilistic choices up to $n$ steps. Observe that the termination of the simulation is guaranteed by the hard limit on the computation length and the binary branching at probabilistic choices.
    
    $M$ then computes the termination probability amassed in the generated execution tree. If this termination probability exceeds $\delta$, $M$ exists to the outer loop to query a new partial scheduler. Otherwise, $M$ returns $0$. $M$ only returns $1$ after exhausting all partial schedulers at level $n$.

    We state that
    $$
    P \in \nAST \Longleftrightarrow \forall \delta \in \setOfRationals^{(0, 1)} \, \exists n \in \setOfNaturals \cdot M(\delta, n, P) = 1
    $$
    Why? \cref{lemma:finitely-many-marked-nodes} indicates the presence of such an $n$ for all $\delta$ for $P \in \nAST$. If $P \not\in \nAST$, there must be some scheduler under which the termination probability is under $\delta$ for some $\delta < 1$. For such a $\delta$, $M$ will return $0$ for all $n$.
    
    Thus, the above equivalence is correct. This produces a $\Pi^0_2$ formula for $\nAST$, completing the proof.
  \end{proof}
 
The membership of $\BAST$ in $\Sigma^0_2$ can be argued similarly. For some $P \in \BAST$, let the bound on the expected runtime be $n$. \citet{KaminskiKM19} showed that in the deterministic case, for each rational $n' < n$, there must be a finite segment of the expected runtime series from the initial program state $\initialProgramState{P}$ that converges to a rational greater than $n'$. With nondeterminism, this finite segment corresponds to another set of marked nodes in the scheduler tree. It turns out that this set of marked nodes is also finite; we leave the details to the diligent reader. This finiteness immediately yields a $\Sigma^0_2$ characterization of $\BAST$.

\section{Transformation to \nameForNormalFormPrograms form}
\label{sec:gambler-transform}

In this section, we informally detail an algorithmic $\PAST$-preserving transformation of $\pGCL$ programs to \nameForNormalFormPrograms form. Our algorithm constructs a program $P_k$ in \nameForNormalFormPrograms form from an input $\mathsf{pGCL}$ program $P$ such that $P_k$ is $\mathsf{PAST}$ \emph{iff} $P$ is $\mathsf{PAST}$.

$P_k$ consists of two components. The first is a program $P^1_k$ that constructs the execution tree of $P$ (see \cref{def:execution-tree})  by simultaneously simulating \emph{all} probabilistic choices. However, $P^1_k$ includes no resolution mechanism for the non-deterministic choices; it leaves them to the scheduler. The goal of $P^1_k$ is to compute finite segments of the expected runtime series (see \cref{def:exp-runtime}) for increasing runtimes of the original program $P$. Let $P^1_k$ store this sum in a variable named \codeStyleText{currentExpRuntime}. Separately, $P^1_k$ sets up a target bound over the expected runtime; more simply, it initializes a variable named \codeStyleText{bound} to \codeStyleText{1}. Observe that $P^1_k$ contains no probabilistic instructions.

The second component $P^2_k$ is a probability halver that, using the \nameForNormalFormPrograms operation $\mathtt{skip }\oplus_{1/2}\mathtt{exit}$, halves the probability of continued execution at each execution step of $P^1_k$. Additionally, $P^2_k$ has the ability to ``cheer'' (like \codeRef{fig:simulation}) for long enough to increase the expected runtime of the overall program by a constant amount (say \codeStyleText{1}). It performs this by storing the current execution probability $\frac{1}{2^n}$ and looping for $2^n$ many steps.

The overall program $P_k$ proceeds as follows. The simulation $P^1_k$ takes in inputs at non-deterministic locations from the input scheduler, and continuously updates the variable \codeStyleText{currentExpRuntime}. In parallel to $P^1_k$, $P^2_k$ repeatedly halves the probability of continued execution. When \codeStyleText{currentExpRuntime} exceeds \codeStyleText{bound}, both $P^1_k$ and $P^2_k$ are paused, and the operation \codeStyleText{bound := bound * 2} is executed. After this, the probability halver $P^2_k$ ``cheers'' for the appropriate number of steps to increase the expected runtime of $P_k$ by a constant amount. Once this is accomplished, $P^1_k$ is resumed, and $P^2_k$ again takes up its \nameForNormalFormPrograms duties. This continues until the next time \codeStyleText{currentExpRuntime} exceeds \codeStyleText{bound}.

Observe now that if \codeStyleText{currentExpRuntime} exceeds \codeStyleText{bound} infinitely often under some scheduler $f$, the expected runtime of $P_k$ under $f$ is infinity; this is the effect of cheering infinitely often. Moreover, the fact that \codeStyleText{currentExpRuntime} always eventually exceeds \codeStyleText{bound} for any fixed value of \codeStyleText{bound} indicates that the original program $P$ is not $\mathsf{PAST}$; i.e., there must exist a scheduler under which expected runtime of the original program $P$ must be infinity (see Theorem 4.2 for a detailed presentation of similar arguments).

Now, for every scheduler, if the expected runtime of $P$ was bounded, then for every scheduler of $P_k$, eventually \codeStyleText{bound} exceeds this bound, and the constructed program only cheers finitely often. If this is the case for all schedulers, both the constructed program and the original program are PAST.

\section{Fair Probabilistic Termination}
\label{sec:discussion}

In this section, we detail a few points that augment our main results.

\subsection{Fairness}
\label{subsec:fairness}



We begin with a look at the $\PAST$, $\BAST$, and $\AST$ problems under the restriction of fairness. Unfortunately, we do not have a proper complexity characterization of these problems; we do however observe that this variant of $\AST$ is at least as hard as the general $\nPAST$, indicating a significant jump in difficulty. We begin by defining fairness in the context of probabilistic programs.
\begin{definition}[Strong Fairness] 
    \label{def:strong-fairness}
    Let $P$ be a $\pGCL$ program, and $f$ be one of its schedulers. $f$ is \emph{strongly fair} if in every infinite branch of the execution tree of $P$ under $f$, every direction at each nondeterministic command visited infinitely often along the branch is taken infinitely often in that branch.
\end{definition}
We can think of a nondeterministic direction being \emph{enabled} at an execution state if the top of the program contains a nondeterministic choice operator presenting that direction as an option. In this context, strong fairness merely mandates the infinitely-often choosing of every direction infinitely-often enabled in every branch.

We now describe a \emph{predicate} $\fairnessPredicateWithoutParameters$ that evaluates to $\top$ only for fair executions. 
Over non-probabilistic programs, strong fairness can be expressed as an arithmetical relation:
\begin{align*}
    \fairnessPredicate{P, f} \triangleq \forall d \in \mathbb{D} \cdot ( \forall k_1 \exists k_2 \cdot &k_2 > k_1 \land \mathsf{enabled}(P, f, d, k_2) )  \\
    \to &(\forall k_1 \exists k_2 \cdot k_2 > k_1 \land \mathsf{taken}(P, f, d, k_2))
\end{align*}
Here, $k_1$ and $k_2$ range over $\setOfNaturals$.
The set $\mathbb{D}$ is the (finite) collection of all \emph{nondeterministic directions} available in the program.
The recursive predicates $\mathsf{taken}(P, f, d, k_2)$ and $\mathsf{enabled}(P, f, d, k_2)$ check if the $k^{th}_2$ step of the execution of $P$ under $f$ enables or takes the direction $d$.

Over probabilistic programs, we want to check fairness over \emph{every} branch. Hence, fairness becomes
\begin{equation}
    \label{eq:fairness-predicate-definition}
    \begin{aligned}
        \fairnessPredicate{P, f} \triangleq \sforall b \in \mathbb{B}\; \forall d \in \mathbb{D} \cdot ( \forall k_1 \exists k_2 \cdot &k_2 > k_1 \land \mathsf{enabled}(P, f, b, d, k_2) )  \\
        \to &(\forall k_1 \exists k_2 \cdot k_2 > k_1 \land \mathsf{taken}(P, f, b, d, k_2))
    \end{aligned}
\end{equation}
Or equivalently,
\begin{equation}
    \label{eq:fairness-predicate-short-definition}
    \fairnessPredicate{P, f} \triangleq \sforall b \in \mathbb{B} \cdot \fairnessBranchPredicate{P, f, b}
\end{equation}
Where $\fairnessBranchPredicateWithoutParameters$ is simply the rest of the formula in \cref{eq:fairness-predicate-definition}. Here, $\mathbb{B}$ denotes the set of all possible branches. Observe that there are uncountably many branches in the infinite tree; quantifying over them is not possible in first-order arithmetic. Consequently, in \cref{eq:fairness-predicate-short-definition}, the universal quantifier is second-order.

We now precisely define the fair versions of $\AST$ and $\PAST$.
\begin{definition}[$\;\nfAST\;$]
    \label{def:nfAST}
    The set $\nfAST$ contains precisely the set of all $\pGCL$ programs $P$ that yield a termination probability of $1$ under any fair scheduler $f$, i.e.,
    $$
    \nfAST \triangleq \left\{ P \in \setOfPrograms \mid \sforall f \in \setOfSchedules \cdot \fairnessPredicate{P, f} \to \termProb(\initialProgramState{P}, f) = 1\right\}
    $$
\end{definition}
\begin{definition}[$\;\nfPAST\;$]
    \label{def:nfPAST}
    The set $\nfPAST$ contains precisely the set of all $\pGCL$ programs $P$ that yield a finite expected runtime under any fair scheduler $f$, i.e.,
    $$
    \nfPAST \triangleq \left\{P \in \setOfPrograms \mid \sforall f \in \setOfSchedules\; \exists k \in \setOfNaturals \cdot \fairnessPredicate{P, f} \to \expRuntime(\initialProgramState{P}, f) < k \right\}
    $$
\end{definition}
\begin{definition}[$\;\nfBAST\;$]
    \label{def:nfBAST}
    The set $\nfBAST$ contains precisely the set of all $\pGCL$ programs $P$ with a bound $k$ over the expected runtime under every fair scheduler $f$, i.e.,
    $$
    \nfPAST \triangleq \left\{P \in \setOfPrograms \mid \exists k \in \setOfNaturals\; \sforall f \in \setOfSchedules \cdot \fairnessPredicate{P, f} \to \expRuntime(\initialProgramState{P}, f) < k \right\}
    $$
\end{definition}
An immediate consequence of the $\Pi^1_1$-hardness of fair termination is that $\nfAST$, $\nfPAST$, and $\nfBAST$ are all $\Pi^1_1$-hard.
This is because all these sets contain every fairly terminating non-probabilistic program. Hence,
\begin{theorem}
    \label{lemma:pi11-hardness-nfAST-nfPAST}
    The sets $\nfAST$, $\nfPAST$, and $\nfBAST$ are each $\Pi^1_1$-hard.
\end{theorem}

Unfortunately, the only upper bound we present here puts each of these sets at one level higher in the analytical hierarchy.
\begin{lemma}
    Each of $\nfAST$, $\nfBAST$, and $\nfPAST$ are contained in $\Pi^1_2$.
\end{lemma}
\begin{proof}
    We show this result for $\nfAST$ (the case for the other sets is similar). Recall, from \cref{def:nfAST}, that
    $$
    P \in \nfAST \Longleftrightarrow \sforall f \in \setOfSchedules \cdot \fairnessPredicate{P, f} \to \termProb(\initialProgramState{P}, f) = 1
    $$
    Using \cref{eq:fairness-predicate-short-definition} gives,
    $$
    P \in \nfAST \Longleftrightarrow \sforall f \in \setOfSchedules\; \sexists b \in \mathbb{B} \cdot \lnot \fairnessBranchPredicate{P, f, b} \lor \termProb(\initialProgramState{P}, f) = 1
    $$
    Both $\fairnessBranchPredicateWithoutParameters$ and $\termProb$ can be calculated by an arithmetical turing machine, and both $f$ and $b$ are second-order variables. Hence, $\nfAST \in \Pi^1_2$.
\end{proof}

\subsection{Finitary Fairness}
\label{subsec:finitary-fairness}


We now study the problems $\AST$, $\PAST$, and $\BAST$ under the restriction that all relevant schedulers are finitary fair.
We begin with a necessary definition of bounded schedulers.
\begin{definition}[$k$-bounded Scheduler]
    \label{def:k-bounded-scheduler}
    Let $P$ be a $\pGCL$ program, $f$ be one of its schedulers, and $k \in \setOfNaturals$. $f$ is \emph{$k$-bounded} if in \emph{every} branch of the execution tree of $P$ under $f$, no nondeterministic direction is consecutively ignored for more than $k$ times.
\end{definition}
Here, ``ignoring'' a direction is tantamount to enabling and not taking it. Bounded schedulers formalize the notion that a realistic implementation of a fair scheduler would enforce the fairness requirements at each nondeterministic command within $k$ tries, for some unspecified $k \in \setOfNaturals$. These schedulers do so by bounding the number of repeated choices at nondeterministic locations to $k$. Notice that these restrictions apply to \emph{all} branches, not just the infinite ones.

We can now define Finitary fairness.
\begin{definition}[Finitary Fairness]
    \label{def:finitary-fairness}
    Let $P$ be a $\pGCL$ program, and $f$ be one of its schedulers. $f$ is finitary fair if it is $k$-bounded for some $k \in \setOfNaturals$.
\end{definition}

Before moving onto our proofs, we lift a useful transformation from \citet{AlurH98} that converts programs that terminate under finitary fairness assumptions to totally terminating programs. This transformation introduces one new uninitialized variable \codeStyleText{k} that stores the bound on the scheduler. It then uses a finite collection of new global variables to instrument each nondeterministic operation in the program code with counters that track the consecutive decisions of the scheduler and guards that enforce the $k$-boundedness requirements. The details of the transformation are available in Section 4.2.3 of \citet{AlurH98}. It is easy to extend these operations to $\pGCL$; we leave the details to the reader. We call this the \emph{finitary transformation} of the program.

The purpose of the transformation is to ensure that every scheduler of the transformed program corresponds to a finitary fair scheduler of the original program, and vice versa. This is formalized in the following lemma.
\begin{lemma}
    \label{lemma:finitary-fairness-correspondence}
    Let $P$ be a $\pGCL$ program, $\finitaryTransformation{P}$ be the $\pGCL$ program produced by the transformation listed in Section 4.2.3 of \citet{AlurH98} on $P$, and $\codeStyleMath{k}$ be the variable introduced by the transformation tracking the bound on the scheduler in $\finitaryTransformation{P}$.

    Then, for every finitary fair scheduler $f$ of $P$, there exists a number $n$ and a schedule $f'$ of the program $\codeStyleMath{k \coloneqq n}; \finitaryTransformation{P}$ such that the decisions taken at nondeterministic commands in the execution tree generated by $f$ on $P$ correspond to the decisions taken at nondeterministic commands in the execution tree generated by $f'$ on $\codeStyleMath{k \coloneqq n}; \finitaryTransformation{P}$. 

    Furthermore, for every number $n \in \setOfNaturals$ and every scheduler $f'$ of $\codeStyleMath{k \coloneqq n}; \finitaryTransformation{P}$, there exists a finitary fair scheduler $f$ of $P$ with the same correspondence at nondeterministic commands over the execution trees of $f'$ on $\codeStyleMath{k \coloneqq n}; \finitaryTransformation{P}$ and $f$ on $P$.
\end{lemma}
The proof of \cref{lemma:finitary-fairness-correspondence} is trivial, and by construction. We leave it as an exercise to the diligent reader. 

\subsubsection{Almost-sure termination under finitary fairness}

We begin by defining the set $\nffAST$.
\begin{definition}[$\;\nffAST\;$]
    \label{def:nffAST}
    The set $\nffAST$ is the collection of all $\pGCL$ programs $P$ such that for every finitary fair scheduler $f$ of $P$,
    $$
    \termProb(\initialProgramState{P}, f) = 1
    $$
\end{definition}
The $\Pi^0_2$-hardness of $\nffAST$ trivially follows from the $\Pi^0_2$-completeness of $\AST$.
We now show that $\nffAST$ is $\Pi^0_2$-complete by proving membership in $\Pi^0_2$.
\begin{lemma}
    \label{lemma:nffAST-pi02-membership}
    $\nffAST \in \Pi^0_2$
\end{lemma}
\begin{proof}
    This is a trivial consequence of \cref{lemma:finitary-fairness-correspondence}. Let $P$ be a $\pGCL$ program, $\finitaryTransformation{P}$ be its transformed variant, and \codeStyleText{k} be the variable used in $\finitaryTransformation{P}$ for tracking the bound on the scheduler. We claim that $P \in \nffAST$ \emph{iff} for every $n \in \setOfNaturals$, $\codeStyleMath{k \coloneqq n}; \finitaryTransformation{P} \in \nAST$.

    Fix some $n \in \setOfNaturals$. We know from \cref{lemma:finitary-fairness-correspondence} that every scheduler $f'$ of $\codeStyleMath{k \coloneqq n}; \finitaryTransformation{P}$ corresponds to some finitary fair scheduler $f$ of $P$. It stands to reason that the execution trees generated by $f$ on $P$ and by $f'$ on $\codeStyleMath{k \coloneqq n}; \finitaryTransformation{P}$ yield identical termination probabilities. Hence, $P \in \nffAST \Longleftrightarrow \codeStyleMath{k \coloneqq n}; \finitaryTransformation{P} \in \nAST$. Since this is true for all $n$, we state
    $$
    P \in \nffAST \Longleftrightarrow \forall n \in \setOfNaturals \cdot \codeStyleMath{k \coloneqq n}; \finitaryTransformation{P} \in \nAST
    $$
    Now, take the $\Pi^0_2$ formula proving the membership of $\codeStyleMath{k \coloneqq n}; \finitaryTransformation{P}$ in $\nAST$. Simply attaching a universal quantifier over $n$ completes the proof.
\end{proof}

\begin{corollary}
    $\nffAST$ is $\Pi^0_2$-complete.
\end{corollary}

Again, a similar argument can be made for the $\Sigma^0_2$-completeness of $\nffBAST$. We do not present this here.

\subsubsection{Positive almost-sure termination under finitary fairness}

Again, we begin by formally defining the set $\nffPAST$.
\begin{definition}[$\;\nffPAST\;$]
    \label{def:nffPAST}
    The set $\nffAST$ contains precisely all $\pGCL$ programs $P$ such that for every finitary fair scheduler $f$ of $P$, $\expRuntime(\initialProgramState{P}, f)$ is finite.
\end{definition}
Proving membership of $\nffPAST$ in $\Pi^1_1$ is easy. 
\begin{lemma}
    \label{lemma:nffPAST-pi11-membership}
    $\nffPAST \in \Pi^1_1$
\end{lemma}
\begin{proof}
    The proof goes through much the same way as that of \cref{lemma:nffAST-pi02-membership}. Let $P$ be a $\pGCL$ program, $\finitaryTransformation{P}$ be its transformed variant, and \codeStyleText{k} be the variable used in $\finitaryTransformation{P}$ for tracking the bound on the scheduler. As before, for some $n \in \setOfNaturals$, each scheduler $f'$ of $\codeStyleMath{k \coloneqq n}; \finitaryTransformation{P}$ can be associated with some scheduler $f$ of $P$. We conclude, from \cref{lemma:finitary-fairness-correspondence}, that the execution tree of $\codeStyleMath{k \coloneqq n}; \finitaryTransformation{P}$ yields a finite expected runtime under $f'$ \emph{iff} the execution tree of $P$ yields a finite expected runtime under $f$. Hence,
    $$
    P \in \nffPAST \Longleftrightarrow \forall n \in \setOfNaturals \cdot \codeStyleMath{k \coloneqq n}; \finitaryTransformation{P} \in \nPAST
    $$
    Simply attach the $\Pi^1_1$ formula for $\nPAST$ membership to complete the proof.
\end{proof}
To show that $\nffPAST$ is $\Pi^1_1$-complete, we reduce $\setOfRecursiveWellFoundedTrees$ to $\nffPAST$.
\begin{lemma}
    \label{lemma:nffPAST-pi11-hardness}
    $\nffPAST$ is $\Pi^1_1$-hard.
\end{lemma}
\begin{proof}[Proof (intuition)]
    In our reduction \codeRef{fig:nffpast-hardness-program}, we build a program that's similar to \codeRef{fig:simulation}, without the need for \codeStyleText{numGen} to resolve Case 3 of the lower bound proof. The important differences are highlighted in \textcolor{red}{red}.
    \begin{figure}
        \declareCodeFigure
        \small
    \begin{lstlisting}[captionpos=b, language=python, mathescape=true, escapeinside={<@}{@>}, label={lst:nffpast-hardness-program}, xleftmargin=15pt]
    node, s $\coloneqq$ [], 1
    while (True):
        x, y, z, w, <@\textcolor{red}{\texttt{k}}@> $\coloneqq$ <@\textcolor{red}{-1}@>, 0, 0, 0, <@\textcolor{red}{\texttt{1}}@>
        while (y = 0):
            x $\coloneqq$ x + 1
            y $\coloneqq$ 0 $\nChoice$ y $\coloneqq$ 1 <@\label{line:nChoice-nffpast}@>
            if (y = <@\textcolor{red}{\texttt{k}}@>): break
            <@\textcolor{red}{\texttt{k $\coloneqq$ y}}@>
            skip $\pChoice{1/2}$ exit
            prob $\coloneqq$ s * 2
        <@\textcolor{red}{\texttt{if (x = 0): exit}}@>
        node $\coloneqq$ node.append(x)
        z $\coloneqq$ execute(M, node)
        if (z = 0): exit
        while (w < s):
            w $\coloneqq$ w + 1
    \end{lstlisting}
    \caption{\emph{The program proving $\Pi^1_1$-completeness of $\nffPAST$.}}
    \label{fig:nffpast-hardness-program}
    \end{figure}
    The basic idea is that we \emph{swap} the direction of the nondeterministic operation at \listingLineRef{line:nChoice-nffpast} that guards the continued execution of the loop. Hence, the directions taken need to be constant swapped to reach higher and higher values of $x$, which is necessary in completing the reduction. To avoid edge cases, we initialize \codeStyleText{x} to \codeStyleText{-1}, and immediately exit if the increment is performed precisely once.

    If the machine $M$ does characterize a well-founded tree, \codeRef{fig:nffpast-hardness-program} is $\nPAST$ for all schedulers. If $M$ fails to characterize a well-founded tree, then there is a scheduler bounded by 2 that simulates one of the (possibly many) infinite branches of the recursive tree. We leave the details of this argument to the reader.
\end{proof}
\begin{corollary}
    \label{corollary:nffPAST-completeness}
    $\nffPAST$ is $\Pi^1_1$-complete.
\end{corollary}

\newoutputstream{todos}
\openoutputfile{main.todos.ctr}{todos}
\addtostream{todos}{\arabic{@todonotes@numberoftodonotes}}
\closeoutputstream{todos}

\label{endofdocument}
\newoutputstream{pagestotal}
\openoutputfile{main.pagestotal.ctr}{pagestotal}
\addtostream{pagestotal}{\getpagerefnumber{endofdocument}}
\closeoutputstream{pagestotal}

\end{document}